\documentclass[conference]{IEEEtran}
\IEEEoverridecommandlockouts
\usepackage{amsthm}
\usepackage{balance}
\newtheorem{theorem}{Theorem}
\newtheorem{lemma}{Lemma}

\newtheorem{assumption}{Assumption}

\newtheorem{claim}{Claim}

\usepackage{algorithm}
\usepackage{makecell}
\usepackage{cite}
\usepackage{amsmath,amssymb,amsfonts}
\usepackage{algorithmic}
\usepackage{graphicx}
\usepackage{caption}
\usepackage{subcaption}
\usepackage{textcomp}
\usepackage{xcolor}
\usepackage{mathtools}
\mathtoolsset{showonlyrefs}
\usepackage{dsfont}

\usepackage{pgfplots}
\pgfplotsset{compat=1.18}
\usepackage{tikz}

\allowdisplaybreaks

\def\BibTeX{{\rm B\kern-.05em{\sc i\kern-.025em b}\kern-.08em
    T\kern-.1667em\lower.7ex\hbox{E}\kern-.125emX}}
\begin{document}

\title{Sequential Spectral Clustering of Data Sequences\\
\thanks{This work is supported in part by Qualcomm University Research Grant.}
}

\author{\IEEEauthorblockN{G Dhinesh Chandran}
\IEEEauthorblockA{\textit{Department of Electrical Engineering} \\
\textit{Indian Institute of Technology Madras}\\
Chennai, India \\
ee22d200@smail.iitm.ac.in}
\and
\IEEEauthorblockN{Kota Srinivas Reddy}
\IEEEauthorblockA{\textit{Department of Artificial Intelligence} \\
\textit{Indian Institute of Technology Kharagpur}\\
Kharagpur, India \\
ksreddy@ai.iitkgp.ac.in}
\and
\IEEEauthorblockN{Srikrishna Bhashyam}
\IEEEauthorblockA{\textit{Department of Electrical Engineering} \\
\textit{Indian Institute of Technology Madras}\\
Chennai, India \\
skrishna@ee.iitm.ac.in}
}

\maketitle

\begin{abstract}
We study the problem of non-parametric clustering of data sequences, where each data sequence comprises independent and identically distributed (i.i.d.) samples generated from an unknown distribution. The true clusters are the clusters obtained using the Spectral clustering algorithm (SPEC) on the pairwise distance between the true distributions corresponding to the data sequences. Since the true distributions are unknown, the objective is to estimate the clusters by observing the minimum number of samples from the data sequences, given a specified error probability. 
To solve this problem, we propose the Sequential Spectral clustering algorithm (SEQ-SPEC), and show that it stops in finite time almost surely and is exponentially consistent. 
We also propose a computationally more efficient algorithm called the Incremental Approximate Sequential Spectral clustering algorithm (IA-SEQ-SPEC). 
Through simulations, we show that both SEQ-SPEC and IA-SEQ-SPEC perform better than the fixed sample size SPEC, the Sequential $K$-Medoids clustering algorithm (SEQ-KMED), and the Sequential Single Linkage clustering algorithm (SEQ-SLINK). 
In addition, we propose memory-efficient versions, SEQ-SPEC-B and IA-SEQ-SPEC-B. Unlike other related sequential clustering algorithms, which require storing all past samples, these algorithms require storing only the most recent $B$ samples.
Both the computationally efficient and memory-efficient versions of SEQ-SPEC perform comparably to SEQ-SPEC in simulations.
\end{abstract}


\section{Introduction} \label{sec: intro}
Clustering, a process of dividing a set of items into groups based on the similarity between the items, has numerous applications  \cite{JAIN2010651,maccuish2010clustering,malik2001contour,bach2006learning,yang2024optimal}. Clustering algorithms can be linkage-based, such as the Single Linkage clustering algorithm (SLINK) \cite{rohlf198212}, or cost minimization problems, like $K$-Means \cite{ahmed2020k} and $K$-Medoids \cite{kaur2014k}, or the Spectral clustering algorithm (SPEC) \cite{ng2001spectral}. The spectral clustering procedure arises as a natural solution to the graph partitioning problem, where the objective is to minimize the normalized cut index. Spectral clustering has advantages over other methods in many applications, including computer vision \cite{malik2001contour} and speech separation \cite{bach2006learning}.

We focus on the class of clustering problems in which we group the data sequences drawn from unknown distributions into clusters based on the distances between the underlying distributions. Data sequence clustering can be studied either under a Fixed Sample Size (FSS) setting or in a Sequential (SEQ) setting. In the FSS setting, the number of samples in each data sequence is fixed, and the sequences are available as a batch. A distance metric is typically used to measure the dissimilarity between the distributions of two data sequences, after which traditional data point clustering algorithms are applied to group the data sequences. Clustering algorithms such as $K$-Medoids \cite{wang2019k} and SLINK \cite{wang2020exponentially} have been studied in this setting. In the SEQ setting, samples are observed sequentially from each of the data sequences. The sequential clustering algorithm is expected to stop at some time and output the correct clusters. The stopping rule typically minimizes the number of samples for a given error probability requirement. At any given time in the SEQ setting, the problem can be viewed as an instance of the FSS setting; thus, traditional data point clustering methods can be applied at each time, supplemented with an appropriate stopping rule. Sequential variants of  $K$-Medoids \cite{sreenivasan2023nonparametric} and SLINK \cite{singh2025exponentially} have been studied. 
Another class of clustering problems involves observing samples from a single data stream, where the objective is to cluster the observed data points. This problem, referred to as Online Clustering \cite{liberty2016algorithm}, differs from our sequential clustering setting.

In this work, we first propose SEQ-SPEC, the sequential version of the Spectral clustering algorithm (SPEC). Our analysis shows that SEQ-SPEC stops in finite time almost surely and is exponentially consistent, i.e., the probability of error decays exponentially with the expected stopping time of the algorithm. Building upon efficient spectral clustering methods for evolving graphs \cite{dhanjal2014efficient, deng2024efficient}, we also propose a more computationally efficient SEQ-SPEC algorithm, called IA-SEQ-SPEC, tailored to our framework. IA-SEQ-SPEC takes advantage of the computations in the previous time instants to reduce complexity.
Simulation results show that both SEQ-SPEC and IA-SEQ-SPEC perform better than the Sequential SLINK algorithm (SEQ-SLINK) \cite{singh2025exponentially}, the Sequential $K$-Medoids algorithm (SEQ-KMED) \cite{sreenivasan2023nonparametric}, and FSS Spectral clustering (SPEC). 
In addition, we propose memory-efficient algorithms, SEQ-SPEC-B and IA-SEQ-SPEC-B, which use a memory-efficient estimator to estimate the pairwise distance between two distributions, as introduced in \cite{gretton2012kernel}, where the kernel two-sample test is studied.
Unlike a classical estimator, which requires storing all samples to update the estimate upon observing a new sample, the memory-efficient estimator in \cite{zaremba2013b} requires storing only the most recent $B$ samples.
Through simulations, we show that all the proposed efficient variants of SEQ-SPEC perform comparably to SEQ-SPEC.

\section{Problem Setup and Preliminaries} \label{sec: problem}

    We consider a collection of $M$ data sequences $\big\{X^{(i)}, i \in [M]=\{1, \ldots, M\}\big\}$, where the $i^{th}$ data sequence is a sequence of i.i.d. multidimensional samples generated from the probability distribution $p_i$. We call the collection of these probability distributions $P = \left\{ p_i, i \in [M] \right\}$ as a problem instance. These $M$ data sequences form $K$ clusters, and it is represented by the set $D = \left\{D_k, k \in [K]\right\}$, which we call a configuration, where $D_k$ is the set of indices of the data sequences that belong to the $k^{th}$ cluster. 
    
    We use Maximum Mean Discrepancy (MMD)\footnote{The analysis can be easily extended to Kolmogorov-Smirnov distance (KSD) by appropriately using the KSD concentration bound in \cite{gretton2012kernel}.} to measure the distance between two distributions.
   Let $d_{i,j}$ denote the MMD distance between the distributions $p_i$ and $p_j$, defined as $d_{i,j} \coloneqq \sup_{f \in \mathcal{F}} \left( \mathbb{E}_{p_i}[f(X)] - \mathbb{E}_{p_j}[f(Y)] \right)$,
where $X \sim p_i$ and $Y \sim p_j$.
 As in \cite{gretton2012kernel}, $\mathcal{F}$ is the unit ball in the reproducing kernel Hilbert space. We use $A$ to denote the affinity matrix corresponding to problem instance $P$, and the $(i, j)^{th}$ entry of $A$ is $A_{ij}=\exp\left(-d_{ij}^2/2\sigma_a^2\right)\mathds{1}\{i \neq j\}$, for some $\sigma_a>0$.
   We use $D$ to denote the degree matrix, which is a diagonal matrix, whose $i^{th}$ diagonal entry is defined as $D_{ii} \coloneqq \sum_{j=1}^MA_{ij}$. We use $L$ to denote the Lagrange matrix and is defined as $L = D^{-\frac{1}{2}}AD^{-\frac{1}{2}}$. Let us define the matrix $Z \in \mathbb{R}^{M\times K}$, whose columns are the $K$ orthonormal eigen vectors corresponding to the top $K$ highest eigen values of the Lagrange matrix $L$. We define the matrix $Y$, whose $i^{th}$ row $Y_i$ is the normalized $i^{th}$ row of $Z$, that is, $Y_i = \frac{Z_i}{\|Z_i\|}$, where $Z_i$ is the $i^{th}$ row of the matrix $Z$. We call $Y_i$ as the spectral point corresponding to the $i^{th}$ data sequence \cite{ng2001spectral}. 

     Given any problem instance $P$, the configuration $D$ is not arbitrary; instead, we can find its corresponding configuration $D$ by using the Spectral clustering algorithm (SPEC) presented in \cite{ng2001spectral} on the affinity matrix $A$ defined above. SPEC applies the K-Means algorithm to the spectral points $\{Y_i: i \in [M]\}$ to determine the configuration.
    In our problem, the problem instance $P$ is unknown; instead, we observe a sequence of multidimensional samples from each of the $M$ data sequences. 
    Our objective is to design a sequential algorithm that groups these $M$ data sequences into $K$ clusters using as few samples as possible for a given error probability. 

\section{Sequential Spectral clustering Algorithm (SEQ-SPEC)} \label{sec: SEQ-SPEC}
The proposed Sequential Spectral clustering algorithm (SEQ-SPEC) is presented in Algorithm \ref{Algo:SEQ-SPEC}. We use $X^{(i)}_t$ to denote the multidimensional sample obtained at time $t$ from the $i^{th}$ data sequence. We use the notations introduced in Section \ref{sec: problem}, but with a hat ($\hat{\cdot}$) and the index $(t)$, to indicate that they are the estimated quantities using $t$ samples, i.e., we use $\hat{A}(t)$, $\hat{Y}_i(t)$, $\hat{D}_k(t)$, etc. Let $k(x, y)$ be the kernel function\footnote{We assume that the kernel function is bounded, i.e., $0\leq k(x, y)\leq B$. In our simulations, we use a Gaussian kernel, i.e., $k(x, y)\coloneqq \exp\left( -\|x-y\|^2/2\sigma_g^2 \right)$ for some $\sigma_g>0$.}. The biased MMD estimate was proposed in \cite{gretton2012kernel}, and its efficient recursive update version in \cite{sreenivasan2023nonparametric}, for $t\geq 0$, with $\hat{d}_{ij}(0) = 0$, is:  
\begin{equation} \label{eq: mmdupdate}
\begin{aligned}
    &\hat{d}_{ij}(t+1) = \frac{1}{t+1}\left\{ \sum_{l=1}^{t+1} h\left( X_l^{(i)}, X_{t+1}^{(i)}, X_l^{(j)}, X_{t+1}^{(j)} \right)  \right.\\ 
    &\left.+\sum_{l=1}^{t} h\left( X_{t+1}^{(i)}, X_l^{(i)}, X_{t+1}^{(j)}, X_l^{(j)} \right) + t^2 \hat{d}_{ij}(t) \right\}^{\frac{1}{2}},
\end{aligned}
\end{equation}
 where $h(x_1, x_2, y_1, y_2) = k(x_1, x_2) + k(y_1, y_2) - 2k(x_1, y_2)$.

\begin{algorithm} 
\caption{SEQ-SPEC}\label{Algo:SEQ-SPEC}
\begin{algorithmic}[1]
\STATE \textbf{Input:}  $K$
\STATE \textbf{Initialize:} $t=0$, $\hat{d}_{ij}(0) = 0, i<j, \forall i, j \in [M]$

\REPEAT
\STATE Get a sample from each of $M$ data sequences. Set $t\leftarrow t+1$ and update all the pairwise MMD estimates, $\left\{ \hat{d}_{ij}(t), i<j, \forall i, j \in [M] \right\}$ using equation \eqref{eq: mmdupdate}.
\STATE Compute $\hat{A}(t)$, with $\hat{A}_{ij} = e^{-\hat{d}_{ij}^2/2\sigma_a^2} \mathds{1}\{i \neq j\}$.
\STATE Use SPEC to get the spectral points $\{\hat{Y}_i(t), \forall i \in [M]\}$ and the clusters $\{\hat{D}_k(t),\forall k \in [K]\}$
\STATE Compute $\displaystyle \Gamma_t=\min_{k\neq l} \min_{\substack{i \in \hat{D}_k(t), j \in \hat{D}_l(t)}}  \left\| \hat{Y}_i(t) - \hat{Y}_j(t) \right\|$.
\UNTIL{$\Gamma_t \geq \arcsin\left( \frac{C}{\sqrt{t}} \right)$} 
\STATE \textbf{Output:} $\{\hat{D}_k(t), \forall k \in [K]\}$
\end{algorithmic}
\end{algorithm}
The inputs to SEQ-SPEC is the number of clusters $K$. 
At each time step $t$, get a sample from each of the data sequences and update the MMD estimates of all pairwise distances (Line 4 of Algorithm \ref{Algo:SEQ-SPEC}). Compute the Affinity matrix and use SPEC to find the spectral points and the estimated clusters (Lines 5 and 6). Compute the estimated minimum inter-cluster distance $\Gamma_t$ and stop sampling the data sequences if $\Gamma_t$ exceeds the threshold\footnote{Sequential algorithms based on $K$-Medoids \cite{sreenivasan2023nonparametric} and SLINK \cite{singh2025exponentially} use the stopping threshold $\frac{C}{\sqrt{t}}$. Since our method computes the inter-cluster distance in the spectral domain, and due to the influence of Davis--Kahan's symmetric $\sin\theta$ theorem \cite{davis1970rotation} in our analysis, we employ the threshold $\arcsin\!\left(C/\sqrt{t}\right)$. However, in simulations, even using $C/\sqrt{t}$ yields similar performance.} $\arcsin\left(C/\sqrt{t}\right)$, otherwise proceed to the next time step $t+1$ (Lines 7 and 8). Once the algorithm stops, it outputs the latest estimated clusters (Line 9). 

\section{Performance Analysis of SEQ-SPEC} \label{sec:perform}
First, we show that FSS-SPEC is exponentially consistent (Theorem \ref{theorem: expconsfss}). Then, we show that SEQ-SPEC stops in finite time almost surely (Theorem \ref{theorem: finitestoptime}) and SEQ-SPEC is universally exponentially consistent (Theorem \ref{theorem: expconsSEQ}). 
The proofs are presented in the Technical Appendix.

We start with the Assumptions A1.1, A2, A3, A4 on the underlying cluster separation in \cite{ng2001spectral}. 
Define $d_i^{(k)} \coloneqq \sum_{j \in D_k} A_{ij}$. The assumptions are:
\begin{assumption} \label{ass: 1}
    Define $h(D_k)$ for all $k \in [K]$ as $h(D_k) \coloneqq \min_{\mathcal{I}\subset D_k} \frac{\sum_{i \in \mathcal{I}, j \in D_k \setminus\mathcal{I}}A_{i, j}}{\min\left\{ \sum_{i \in \mathcal{I}}d_i^{(k)}, \sum_{j \in D_k \setminus \mathcal{I}}d_j^{(k)}  \right\}}$.
    There exist $\delta>0$ such that $\frac{\left(h(D_k)\right)^2}{2} \geq \delta$ for all $k \in [K]$.
\end{assumption}

\begin{assumption} \label{ass: 2}
    There exists a $\epsilon_1>0$ such that for all $k_1, k_2 \in [K], k_1\neq k_2$, we have $\sum_{i \in D_{k_1}} \sum_{j \in D_{k_2}} \frac{A_{i, j}^{2}}{d_i^{(k_1)}d_j^{(k_2)}} \leq \epsilon_1$.
\end{assumption}

\begin{assumption} \label{ass: 3}
    There exists a $\epsilon_2>0$ such that $\forall k \in [K]$, $i\in D_k$, we have $\left[ \sum_{j \notin D_k} \frac{A_{i, j}}{d_i^{(k)}} \right] \left[ \sum_{i_1, j_1 \in D_k}\frac{A_{i_1, i_2}^{2}}{d_{i_1}^{(k)} d_{j_1}^{(k)}} \right]^{\frac{1}{2}} \leq \epsilon_2$.
\end{assumption}

\begin{assumption} \label{ass: 4}
    There exists a $C>0$ such that for all $k \in D_k$, $i \in D_k$, we have $\frac{1}{M_k}\frac{\sum_{i_1 \in D_k}d_{i_1}^{(k)}}{d_i^{(k)}} \leq C$, 
    where $M_k$ is the number of data sequences in the $k^{th}$ cluster.
\end{assumption}

The intuition behind these assumptions is discussed in detail in \cite{ng2001spectral}. In short, Assumption \ref{ass: 1} ensures that each cluster is hard to split into two dissimilar partitions. Assumptions \ref{ass: 2} and \ref{ass: 3} ensure that the nodes within clusters are more similar relative to nodes between clusters. Assumption \ref{ass: 4} ensures that no node in a cluster is too much less connected relative to other nodes in the same cluster.


Let $E_t$ denotes the error event at time $t$, i.e., $E_t \coloneqq \left\{ \exists k \in [K], \hat{D}_k(t) \neq D_k\right\}$. 
Define the event $F_{\epsilon,t}\coloneqq \left\{ \forall i, j \in [M], i<j,  \left|\hat{d}_{ij}(t)-d_{ij}\right|<\epsilon\right\}$, where $F_{\epsilon,t}$ denotes the event that all pairwise MMD estimates at time $t$ deviate from their true MMD values by at most $\epsilon$.

\begin{theorem} \label{theorem: expconsfss}
    There exists a $t_0 \in \mathbb{N}$, such that for all $t>t_0$, the error probability of the FSS SPEC with $t$ samples from each of the data sequences is upper bounded as $\mathbb{P}\left[ E_t \right] \leq M^2\exp\left( -\alpha_0 t \right)$ for some $\alpha_0>0$.
\end{theorem}
\begin{proof}[Proof sketch]
Assume that $\delta > (2+\sqrt{2})\boldsymbol{\epsilon}$ and
$CM(4+2\sqrt{K})^2 < r$, with $r < \frac{1}{2\sqrt{2}}$, where
$\boldsymbol{\epsilon} \coloneqq \sqrt{K(K-1)\epsilon_1 + K\epsilon_2^2}$.
By Theorem~2 in~\cite{ng2001spectral}, which quantifies the separation of clusters in the spectral domain, spectral clustering (SPEC) applied to the affinity matrix $A$ outputs a unique clustering under these assumptions.
Hence, at any time $t$, if the corresponding estimated quantities satisfy
\begin{equation} \label{eq:pf1}
    \hat{\delta}(t) > (2+\sqrt{2})\hat{\boldsymbol{\epsilon}}(t)
    \quad \text{and} \quad
    \hat{C}(t)M(4+2\sqrt{K})^2 < r,
\end{equation}
then the estimated clustering at time $t$ is correct.

We show that, under the event $F_{\epsilon,t}$, the conditions in
\eqref{eq:pf1} are satisfied for all $\epsilon < \epsilon_0$, for some
$\epsilon_0 > 0$. 
Consequently, the error probability satisfies $\mathbb{P}\!\left[E_t\right] \leq \mathbb{P}\!\left[F_{\epsilon,t}^{C}\right].$

Since the probability that the MMD estimates deviate from their true values by $\epsilon$, $\mathbb{P}\!\left[F_{\epsilon,t}^{C}\right]$, decays exponentially in $t$ \cite{gretton2012kernel}, it follows that $\mathbb{P}[E_t]$ also decays exponentially with $t$.
\end{proof}

\begin{theorem} \label{theorem: finitestoptime}
    $\mathbb{P}[N<\infty]=1$, where $N$ is the stopping time of SEQ-SPEC.
\end{theorem}
\begin{proof}[Proof sketch]
The event $N > t$ implies that the algorithm does not stop at time $t$, i.e.,
$\Gamma_t < \arcsin\left(C/\sqrt{t}\right)$. From the definition of $\Gamma_t$, this implies that
\[
\|\hat{Y}_i(t) - \hat{Y}_j(t)\| < \arcsin\left( C/\sqrt{t} \right)
\quad \text{for some } i \neq j,
\]
where $i$ and $j$ are data sequences from different clusters.
We argue that this implies that 
the estimated spectral point $\hat{Y}_i(t)$ deviates from the true spectral point $Y_i$ by at least $\delta$, i.e., 
\begin{equation} \label{eq:pf2}
    \|\hat{Y}_i(t) - Y_i\| > \delta \quad \text{for some } \delta>0.
\end{equation}
Now, under the event $F_{\epsilon,t}$, we show that
\begin{equation}\label{eq: pf3}
\|\hat{Y}_i(t) - Y_i\| < A_1 \arcsin(\alpha_1 \epsilon),
\end{equation}
where the bound follows from Proposition~6.1 (the symmetric $\sin\theta$ theorem) in~\cite{davis1970rotation}, which provides a quantitative bound on the deviation of canonical angles between eigen spaces under matrix perturbations, Theorem~5.2 in~\cite{stewart1990matrix}, which relates canonical angles between subspaces to angles between eigenvectors, and Corollary~4.3.15 (Weyl's theorem) in~\cite{horn2012matrix}, which bounds the perturbation of eigenvalues.

Since the probability of the complement event $F_{\epsilon,t}^{C}$ decays exponentially with $t$~\cite{gretton2012kernel}, from \eqref{eq: pf3}, the probability of the event in~\eqref{eq:pf2} also decays exponentially with $t$. Consequently, $\mathbb{P}[N > t]$ decays exponentially with $t$, which proves that $\lim_{t \to \infty} \mathbb{P}[N > t] = 0$, implying that $\mathbb{P}[N<\infty]=1$.
\end{proof}

\begin{theorem} \label{theorem: expconsSEQ}
    There exists a constant $G>0$ such that, as $C\rightarrow\infty$, $\mathbb{E}[N] \leq -G \log{\left(\mathbb{P}\left[E_N\right]\right)}(1+o(1))$. 
\end{theorem}
\begin{proof}[Proof sketch]
    First, using a procedure similar to that of Theorem~\ref{theorem: finitestoptime}, we show that the error probability $\mathbb{P}[E_N]$ decays exponentially in $C^2$. Intuitively, choosing a larger value of the parameter $C$ in SEQ-SPEC leads to a larger stopping time $N$. We formalize this intuition by proving that the ratio of the stopping time $N$ to $C^2$ converges to a problem-dependent constant. Combining these results, we conclude that the error probability $\mathbb{P}[E_N]$ decays exponentially with the expected stopping time $\mathbb{E}[N]$.
\end{proof}



\section{Incremental Approximate Sequential Spectral clustering (IA-SEQ-SPEC)} \label{sec: IA-SEQ-SPEC}

In SEQ-SPEC, at each time $t$, we do an eigen decomposition of the Lagrange matrix $\hat{L}(t)$. Now, we propose a computationally efficient algorithm IA-SEQ-SPEC, which makes use of the eigen decomposition of the Lagrange matrix at the previous time $\hat{L}(t-1)$ to compute the approximate eigen decomposition of Lagrange matrix at the current time $t$, but by updating only small patch of the Affinity matrix $\hat{A}(t)$. This approach builds on \cite{dhanjal2014efficient}, where the Incremental Approximate Spectral (IA-SPEC) clustering algorithm is proposed for the sequence of changing graphs. 
IA-SEQ-SPEC differs from IA-SPEC in the following aspects. (1) IA-SPEC assumes only a certain number of edge changes in the graph at each time, whereas, in our framework, all the edges change at each time $t$. Hence, we follow a sequence of procedures to identify the edges that change the most and update only those edges. (2) In IA-SPEC, the rank of the approximation $l$ is fixed, which may not be a good approximation at all $t$, leading to error propagation. Therefore, we use an adaptive choice of $l$ obtained by upper-bounding the fractional approximation error. (3) IA-SPEC cannot be directly applied to our sequential framework as it lacks the stopping condition. We use a stopping condition based on our analysis of SEQ-SPEC.

Now we discuss the procedure of IA-SEQ-SPEC at time $t$. We compute the difference in the affinity matrix $\Delta(t)=\hat{A}(t)-\Tilde{A}(t-1)$, where $\Tilde{A}(t-1)$ is the modified affinity matrix updated at $t-1$. We choose a $p\times p$ symmetric block $H_{22}$ from the matrix $\Delta(t)$ with the highest absolute sum. For example, if $p=2$, then $H_{22}=\begin{bmatrix}
    \Delta_{ii}(t) & \Delta_{ij}(t) \\
    \Delta_{ji}(t) & \Delta_{jj}(t)
\end{bmatrix}$, for some $i\neq j, i, j \in [M]$ is a valid $p\times p$ symmetric block. Let $\Tilde{\Delta}(t)$ be the matrix obtained from $\Delta(t)$ by setting the entries of the locations other than the chosen $p\times p$ block to be zero. Without loss of generality, assume that the $p\times p$ block is in the bottom right corner of the matrix $\Tilde{\Delta}(t)$, i.e., $\Tilde{\Delta}(t) = \begin{bmatrix}
    \boldsymbol{0}_{M-p\times M-p} & \boldsymbol{0}_{M-p\times p} \\
    \boldsymbol{0}_{p\times M-p} & H_{22}
\end{bmatrix}$, where $H_{22} \in \mathbb{R}^{p\times p}$. Then, we form the modified affinity matrix $\Tilde{A}(t) = \Tilde{A}(t-1) + \Tilde{\Delta}(t)$ and the Lagrange matrix $\hat{L}(t)$ from $\Tilde{A}(t)$. It can be verified that the change in the Lagrange matrix $U(t) = \hat{L}(t) - \hat{L}(t-1)$ takes the form $\begin{bmatrix}
    \boldsymbol{0}_{M-p\times M-p} & U_{12} \\
    U_{21} & U_{22}
\end{bmatrix}$, where $U_{12}, U_{21}^T \in \mathbb{R}^{M-p\times p}$ and $U_{22} \in \mathbb{R}^{p \times p}$.
We use $Q_l(\cdot)$ and $\Omega_l(\cdot)$ to denote the top $l$ eigen vectors and eigen values of $\hat{L}(\cdot)$, respectively. We choose the rank of approximation $l(t)$ such that the ratio of the sum of squares of top $l(t)$ eigen values to the sum of squares of all eigen values of $\hat{L}(t-1)$ is lower bounded by $q$, for some fixed $q\in (0, 1)$.
In our method, instead of finding the eigen decomposition of $\hat{L}(t) = \hat{L}(t-1) + U(t)$, we find the eigen decomposition of ($l(t)$ rank approximation of $\hat{L}(t-1)$) $+$ $U(t)$, i.e., $Q_{l(t)}(t-1)\Omega_{l(t)}(t-1)Q_{l(t)}^T(t-1) + U(t)$.
We find this eigen decomposition, i.e., $Q(t)\Omega(t)Q^T(t)$, using the procedure presented in Section 4.1 in \cite{dhanjal2014efficient}. The order of complexity is $(l^2+p^2)(l+p)+Mp(l+p)$, which is less than the complexity order $M^3$ for the actual eigen decomposition. Once we have the eigen decomposition, we follow the remaining procedure in SPEC described in Section \ref{sec: intro} to get the clusters. 

To avoid the propagation of approximation error, we calculate the exact decomposition once every $R$ time steps. Furthermore, if at some time $t$, the approximate decomposition is followed and the algorithm satisfies the stopping condition in Line 8 of Algorithm \ref{Algo:SEQ-SPEC}, then the stopping condition is again checked by performing the exact eigen decomposition. 
The algorithm stops only if the stopping condition is satisfied with the exact decomposition.

\section{Memory Efficient Sequential Spectral Clustering (SEQ-SPEC-B/IA-SEQ-SPEC-B)} \label{sec: memeff}
We observe that both SEQ-SPEC and IA-SPEC require storing all observed samples up to time $t$ in order to update the pairwise MMD estimates (Line~4 in Algorithm~\ref{Algo:SEQ-SPEC}). Consequently, the storage requirement grows linearly with $t$.
Hence, we propose memory-efficient versions of SEQ-SPEC and IA-SEQ-SPEC, which we call SEQ-SPEC-B and IA-SEQ-SPEC-B, respectively, where we use the pairwise MMD-B estimates proposed in \cite{zaremba2013b}, denoted by $\hat{d}_{ij}^{(B)}(t)$), and require storing at the most recent $B$ samples. To compute the square of the pairwise MMD-B estimate at time $t$, we divide the $t$ samples from each data sequence into blocks of size $B$ and take the average of the square of the pairwise MMD estimates over these blocks.
We use $n(t)=\left\lfloor\frac{t}{B}\right\rfloor$ to denote the number of non-overlapping blocks of size $B$ till time $t$ and $\hat{d}_{ij}(s_1:s_2)$ to denote the pairwise MMD estimate between data sequences $i$ and $j$ using samples from time $s_1$ to $s_2$. The pairwise MMD-B estimates $\hat{d}^{(B)}_{ij}(t)$ is given by, 
\begin{equation}
\begin{aligned}
    \hat{d}_{ij}^{(B)}(t) = &\frac{1}{\sqrt{t}}\left[ \sum_{s=0}^{n(t)-1}B\hat{d}_{ij}^2(sB+1:(s+1)B)\right.\\
    &\hspace{2cm}\left.+\left(t-n(t)B\right)\hat{d}_{ij}^2\left(n(t)B+1:t\right) \right]^{\frac{1}{2}}
\end{aligned}
\end{equation}
Here, we take the weighted average as the first $n(t)$ blocks have $B$ samples and the $n(t)+1^{\text{th}}$ block has $t-n(t)B$ samples.

\section{Simulation Results} \label{sec: Simulations}

\begin{figure*}[htb]
    \centering
    \begin{minipage}[t]{0.30\textwidth}
        \centering
        \begin{tikzpicture}
        \begin{axis}[
            xlabel = {\small{Empirical stopping time}}, 
            ylabel = {$\log\left(\mathbb{P}[E_N]\right)$},
            xmin =-0.5, xmax=152,
            ymin=-6.1, ymax=0.1,
            xtick={0, 25, 50, 75, 100, 125, 150},
            xticklabels={0, 25, 50, 75, 100, 125, 150},
            grid=major,
            ytick = {-6, -5, -4, -3, -2, -1, 0},
            yticklabels = {-6, -5, -4, -3, -2, -1, 0},
            legend style={at={(0, 0)}, draw = none, fill opacity = 0, text opacity = 0.7,anchor=south west, nodes={scale=0.66},legend cell align=left},
            width = 1.0\textwidth, 
            height = 1.0\textwidth, 
            tick label style={font=\tiny},
            enlargelimits=false,
        ]
        \addplot[mark = x,red,thick, mark options={fill=red}] coordinates{
        (10, 0)
        (20, -0.00280392733)
        (40, -0.097392346)
        (60, -0.426944140)
        (80, -1.00731013)
        (100, -1.71981097)
        (120, -2.51207232)
        (140, -3.46733718)
        (160, -4.38202663)
        (180, -5.38169898)
        (200, -6.37712703)
        };
        \addplot[mark=x, blue, thick, mark options={fill=blue}] coordinates{
        (1.9971, 0.00000000)
        (5.4727, 0.00000000)
        (11.4527, -0.000400080021)
        (20.7124, -0.00833463679)
        (33.6378, -0.0713885961)
        (49.2096, -0.318141280)
        (65.4864, -0.814411268)
        (80.9248, -1.58572139)
        (96.1848, -2.58494800)
        (112.8116, -3.86800612)
        (131.6879, -5.40367788)
        (153.2404, -6.90775528)
        };
        \addplot[mark=x, magenta, thick, mark options={fill=magenta}] coordinates{
        (1.1361, 0.00000000)
        (5.9918, 0.00000000)
        (14.3766, -0.000700245114)
        (26.3095, -0.0136933271)
        (41.5041, -0.143985851)
        (59.4288, -0.520202786)
        (78.5124, -1.24966776)
        (98.3514, -2.28278247)
        (118.9352, -3.49002860)
        (141.0478, -4.90627528)
        (164.8898, -6.64539101)
        };
        \addplot[mark=x, brown, thick, mark options={fill=purple}] coordinates{
        (48.1207, -0.2450058)
        (54.8650, -0.33757284)
        (59.6150, -0.44192165)
        (65.6050, -0.5839341)
        (74.3250, -0.81803036)
        (87.0300, -1.26833356)
        (107.3600, -2.66499071)
        (111.0750, -3.09444825)
        (117.6750, -3.75930192)
        (151.8750, -5.68397985)
        (153.7550, -6.72543372)
        (175.5450, -8.11172808)
        (200.7950, -9.21034037)
        };
        \addplot[mark=x, black, thick] coordinates{
        (2.1604417, 0.00000000)
        (6.84627043, -0.000108265861)
        (16.0580275, -0.00597211161)
        (31.56371116, -0.0790418015)
        (51.3683014, -0.337208677)
        (71.61069611, -0.862497047)
        (91.65768107, -1.78238491)
        (104.97455884, -2.74241103)
        (115.17332467, -3.69725043)
        (130.20179712, -5.12363925)
        (147.67034752, -6.42292224)
        (159.7894338, -9.13097243)
        };

\addplot[
    mark=x, green!50!black, thick
]
coordinates {
    (5.46890802,  -0.00520536)
    (16.44526445, -0.01900085)
    (34.47997813, -0.04354215)
    (57.99494328, -0.06666695)
    (86.31515648, -0.08711747)
    (118.63072297,-0.10140649)
    (155.10113434,-0.11093222)
    (195.80388137,-0.11798055)
};
        
        \legend{SPEC, SEQ-SPEC, SEQ-SLINK, SEQ-SPEC-R=50, IA-SEQ-SPEC-R=50, SEQ-KMED}
        \end{axis}
        \end{tikzpicture}
        \vspace{-1mm}
        \caption{Synthetic dataset 1}
        \label{fig:circle}
    \end{minipage}\hfill
    \begin{minipage}[t]{0.30\textwidth}
        \centering
        \begin{tikzpicture}
        \begin{axis}
        [ylabel = {$\log\left(\mathbb{P}[E_N]\right)$}, 
        xlabel = {\small{Empirical Stopping time}}, 
        xmin =0.1, 
        xmax=27,
        ymin=-4.0, 
        ymax=0.1,
        xtick={0, 5, 10, 15, 20, 25},
        xticklabels={0, 5, 10, 15, 20, 25},
        grid=major,
        ytick = {-5, -4, -3, -2, -1, 0},
        yticklabels = {-5, -4, -3, -2, -1, 0},
        legend style={at={(1,0.46)},
        fill opacity = 0.8,
        anchor=south east ,
        nodes={scale=0.66},legend cell align=left},
        width = 1.0\textwidth,
        height = 1.0\textwidth,
        tick label style = {font=\tiny},
        enlargelimits=false,]

        \addplot[mark=x, magenta, thick] coordinates{
        (4.07624426, -0.19654149)
        (14.49558772, -0.34544282)
        (29.08171550, -0.53091068)
        (56.36057183, -0.88390522)
        (72.77214966, -1.10694198)
        (90.51376633, -1.35194581)
        (119.37380868, -1.68569323)
        };
        \addplot[mark=x, blue, thick] coordinates{
            (1.0954, -0.05245186)
            (3.8934, -0.90460923)
            (7.3348, -2.15244243)
            (13.5556, -3.49002860)
        };
        \addplot[mark=x, black, thick] coordinates{
        (2.3421, -0.3361)
        (7.4963, -1.8059)
        (13.3761, -2.9834)
        (17.1270, -3.4976)
        };

        \addplot[
    mark=x, green!50!black, thick
] coordinates {
    (1.0111,  -0.10440)
    (1.1467,  -0.14552)
    (2.0109,  -0.45234)
    (4.7417,  -1.59917)
    (11.3998, -3.74911)
};

        \legend{SEQ-SLINK, SEQ-SPEC, IA-SEQ-SPEC, SEQ-KMED}
        \end{axis}
        \end{tikzpicture}
        \vspace{-1mm}
        \caption{Synthetic dataset 2}
        \label{fig:bridge}
    \end{minipage}\hfill
    \begin{minipage}[t]{0.30\textwidth}
        \centering
\begin{tikzpicture}
\begin{axis}[
    width=\linewidth,
    height=\linewidth,
    grid=both,
    xmin=-4.5, xmax=4.5,
    ymin=-6, ymax=4.5,
    xlabel = {\textcolor{white}{dimension $1$}},
    ytick = {-6, -4, -2, 0, 2, 4},
    axis equal,
    legend columns=3,
    legend style={at={(0,0)}, anchor=south west, fill opacity=0.7, nodes={scale=0.66}},
    tick label style={font=\tiny}
]

\addplot+[only marks, mark=*, mark options={draw=red, fill=red}, mark size=1.5pt]
coordinates {
    ( 1.0000,  0.0000)
    ( 0.5000,  0.8660)
    (-0.5000,  0.8660)
    (-1.0000,  0.0000)
    (-0.5000, -0.8660)
    ( 0.5000, -0.8660)
};

\addplot+[only marks, mark=*, mark options={draw=blue, fill=blue}, mark size=1.5pt]
coordinates {
    ( 2.5000,  0.0000)
    ( 2.2520,  1.0839)
    ( 1.5570,  1.9550)
    ( 0.5567,  2.4373)
    (-0.5567,  2.4373)
    (-1.5570,  1.9550)
    (-2.2520,  1.0839)
    (-2.5000,  0.0000)
    (-2.2520, -1.0839)
    (-1.5570, -1.9550)
    (-0.5567, -2.4373)
    ( 0.5567, -2.4373)
    ( 1.5570, -1.9550)
    ( 2.2520, -1.0839)
};

\addplot+[only marks, mark=*, mark options={draw=green!60!black, fill=green!60!black}, mark size=1.5pt]
coordinates {
    ( 4.0000,  0.0000)
    ( 3.8042,  1.2361)
    ( 3.2361,  2.3511)
    ( 2.3511,  3.2361)
    ( 1.2361,  3.8042)
    ( 0.0000,  4.0000)
    (-1.2361,  3.8042)
    (-2.3511,  3.2361)
    (-3.2361,  2.3511)
    (-3.8042,  1.2361)
    (-4.0000,  0.0000)
    (-3.8042, -1.2361)
    (-3.2361, -2.3511)
    (-2.3511, -3.2361)
    (-1.2361, -3.8042)
    ( 0.0000, -4.0000)
    ( 1.2361, -3.8042)
    ( 2.3511, -3.2361)
    ( 3.2361, -2.3511)
    ( 3.8042, -1.2361)
};

\legend{Cluster 1, Cluster 2, Cluster 3}

\end{axis}
\end{tikzpicture}
\vspace{-1mm}
\caption{Synthetic Dataset 3}
\label{fig:circleprob3clus}

    \end{minipage}
\end{figure*}
\begin{figure*}[htb]
    \centering
    \begin{minipage}[t]{0.30\textwidth}
        \centering
        \begin{tikzpicture}
        \begin{axis}[
            xlabel = {\small{Empirical stopping time}}, 
            ylabel = {$\log\left(\mathbb{P}[E_N]\right)$},
            xmin =-0.5, xmax=205,
            ymin=-5, ymax=0.1,
            xtick={0, 40, 80, 120, 160, 200, 240},
            xticklabels={0, 40, 80, 120, 160, 200, 240},
            grid=major,
            ytick = {-6, -5, -4, -3, -2, -1, 0},
            yticklabels = {-6, -5, -4, -3, -2, -1, 0},
            legend style={at={(0, 0)}, draw = none, fill opacity = 0, text opacity = 0.7,anchor=south west, nodes={scale=0.66},legend cell align=left},
            width = 1.0\textwidth, 
            height = 1.0\textwidth, 
            tick label style={font=\tiny},
            enlargelimits=false,
        ]
        
        \addplot[
            mark=x, red, thick
        ]
        coordinates {
            (50,  -0.02164587)
            (100, -0.54544121)
            (150, -1.68811349)
            (200, -3.2018904)
        };
        
        \addplot[mark=x, blue, thick, mark options={fill=blue}] coordinates{
        (2.113, 0.00000000)
        (35.354, -0.0068232)
        (69.1116, -0.207762449)
        (88.7952, -0.529668712)
        (126.6605, -1.76142427)
        (186.5206, -4.25451331)
        };
     
        \addplot[mark=x, magenta, thick, mark options={fill=magenta}] coordinates{
        (35.17346939, -0.00511509985)
        (88.196793, -0.482324112)
        (127.43804665, -1.61089673)
        (187.6516035, -4.82612987)
        };
        
        \addplot[mark=x, black, thick] coordinates{
        (2.64, 0.00000000)
        (51.51, -0.06806473)
        (93.35, -0.59275901)
        (113.22, -1.11048246)
        (146.03, -2.39909599)
        (195.61, -4.91441644)
        };
\addplot[
    mark=x, green!50!black, thick
]
coordinates {
    (52.27693475,  -0.06828858)
    (110.86949924, -0.97931296)
    (145.42336874, -2.19343809)
    (196.36267071, -4.78597529)
};

        \legend{SPEC, SEQ-SPEC, SEQ-SPEC-B, IA-SEQ-SPEC, IA-SEQ-SPEC-B}
        \end{axis}
        \end{tikzpicture}
        \vspace{-1mm}
        \caption{Synthetic dataset 3}
        \label{fig:circle3clus}
    \end{minipage}\hfill
    \begin{minipage}[t]{0.30\textwidth}
        \centering
        \begin{tikzpicture}
        \begin{axis}
        [ylabel = {$\log\left(\mathbb{P}[E_N]\right)$}, 
        xlabel = {\small{Empirical Stopping time}}, 
        xmin =0.1, 
        xmax=16,
        ymin=-4.0, 
        ymax=0.1,
        xtick={0, 5, 10, 15, 20, 25},
        xticklabels={0, 5, 10, 15, 20, 25},
        grid=major,
        ytick = {-5, -4, -3, -2, -1, 0},
        yticklabels = {-5, -4, -3, -2, -1, 0},
        legend style={at={(1,1)},
        fill opacity = 0.8,
        anchor=north east ,
        nodes={scale=0.66},legend cell align=left},
        width = 1.0\textwidth,
        height = 1.0\textwidth,
        tick label style = {font=\tiny},
        enlargelimits=false,]

        \addplot[mark=x, magenta, thick] coordinates{
        (1.0003, -0.28981768)
        (3.182, -0.70238976)
        (4.4112, -0.99074492)
        (9.7717, -2.54848563)
        (12.1611, -3.26491976)
        };

        \addplot[mark=x, blue, thick] coordinates{
            (2.0146, -0.66572658)
            (3.232, -1.12639467)
            (5.2822, -1.87993516)
            (7.6192, -2.74731092)
            (10.631, -3.70500884)
        };

        \addplot[mark=x, black, thick] coordinates{
        (2.009, -0.65123778)
        (3.0666, -1.02359758)
        (4.6093, -1.48324589)
        (7.2109, -2.24715039)
        (10.4875, -3.13269813)
        };
    
        \addplot[
    mark=x, green!50!black, thick
] coordinates {
    (1.0013,  -0.52374202)
    (2.2203,  -0.82805094)
    (3.6839,  -1.26301534)
    (4.8163,  -1.65862816)
    (9.9746, -3.69691163)
};

        \legend{SEQ-SLINK, SEQ-SPEC, IA-SEQ-SPEC, SEQ-KMED}
        \end{axis}
        \end{tikzpicture}
        \vspace{-1mm}
        \caption{Real-world dataset 1}
        \label{fig:movielens}
    \end{minipage}\hfill
    \begin{minipage}[t]{0.30\textwidth}
        \centering
        \begin{tikzpicture}[scale=1]
        \begin{axis}
        [ylabel = {$\log\left(\mathbb{P}[E_N]\right)$}, 
        xlabel = {\small{Empirical Stopping time}}, 
        xmin =0.1, 
        xmax=48,
        ymin=-6.1, 
        ymax=0.1,
        xtick={0, 5, 10, 15, 20, 25, 30, 35, 40, 45},
        xticklabels={0, 5, 10, 15, 20, 25, 30, 35, 40, 45},
        grid=major,
        ytick = {-6, -5, -4, -3, -2, -1, 0},
        yticklabels = {-6, -5, -4, -3, -2, -1, 0},
        legend style={at={(0,0)},
        fill opacity = 0.7,
        anchor=south west,
        nodes={scale=0.66}, legend cell align=left},
        width = 1.0\textwidth,
        height = 1.0\textwidth,
        tick label style = {font=\tiny},
        enlargelimits=false,]

        \addplot[mark=x, magenta, thick] coordinates{
        (1.0081, 0.00000000)
        (4.4660, -0.00954541)
        (14.5031, -0.44145505)
        (23.5857, -1.24444764)
        (33.1677, -2.30358559)
        (39.5946, -3.14888345)
        (46.1962, -3.98998455)
        };

        \addplot[mark=x, blue, thick] coordinates{
        (2.1534, -0.04845522)
        (5.7386, -1.07645948)
        (10.6378, -3.22640409)
        (17.5006, -5.84304454)
        };
        \addplot[mark=x, black, thick] coordinates{
        (2.0029, -0.0031)
        (3.0599, -0.1739)
        (4.5244, -0.5682)
        (6.8895, -1.4503)
        (9.9053, -2.6647)
        (13.5037, -4.0745)
        (17.6751, -5.6840)
        };
        \addplot[
    mark=x, green!50!black, thick
] coordinates {
    (1.5719,  -0.00170)
    (6.5883,  -0.23154)
    (16.9073, -1.39935)
    (25.0994, -2.46969)
    (33.8430, -3.69936)
    (39.9688, -4.72123)
    (46.3383, -5.60212)
};

        \legend{SEQ-SLINK, SEQ-SPEC, IA-SEQ-SPEC, SEQ-KMED}
        \end{axis}
        \end{tikzpicture}
        \vspace{-1mm}
        \caption{Real-world dataset 2}
        \label{fig:MNIST}
    \end{minipage}
\end{figure*}

\begin{figure}[b]
    \centering
    \begin{minipage}{0.48\linewidth}
        \centering
        \begin{tikzpicture}
            \begin{axis}[
                width=\linewidth,
                height=4cm,
                grid=both,
                xlabel={},
                ylabel={},
                xmin = -3,
                xmax = 3,
                ymin = -3.5,
                ymax = 2.5,
                legend columns = 2,
                legend style={at={(0, 0)}, fill opacity = 0.7,anchor=south west, nodes={scale=0.66}},
                tick label style = {font=\tiny}
            ]
            \addplot+[only marks, mark = *, mark options = {draw = red, fill=red}, mark size = 1.5pt] coordinates {
                (1.00000000, 0.00000000)
                (0.809016994, 0.587785252)
                (0.309016994, 0.951056516)
                (-0.309016994, 0.951056516)
                (-0.809016994, 0.587785252)
                (-1.00000000, 0.00000000)
                (-0.809016994, -0.587785252)
                (-0.309016994, -0.951056516)
                (0.309016994, -0.951056516)
                (0.809016994, -0.587785252)
            };
            \addplot+[only marks, mark = *, mark options = {draw = blue, fill=blue}, mark size = 1.5pt] coordinates {
                (2.00000000, 0.00000000)
                (1.90211303, 0.618033989)
                (1.61803399, 1.17557050)
                (1.17557050, 1.61803399)
                (0.618033989, 1.90211303)
                (0.00000000, 2.00000000)
                (-0.618033989, 1.90211303)
                (-1.17557050, 1.61803399)
                (-1.61803399, 1.17557050)
                (-1.90211303, 0.618033989)
                (-2.00000000, 0.00000000)
                (-1.90211303, -0.618033989)
                (-1.61803399, -1.17557050)
                (-1.17557050, -1.61803399)
                (-0.618033989, -1.90211303)
                (0.00000000, -2.00000000)
                (0.618033989, -1.90211303)
                (1.17557050, -1.61803399)
                (1.61803399, -1.17557050)
                (1.90211303, -0.618033989)
            };

            \legend{Cluster 1, Cluster 2}
            \end{axis}
        \end{tikzpicture}
        \caption{Synthetic Dataset 1} 
        \label{fig: circleprob}
    \end{minipage}
    \hfill
    \begin{minipage}{0.48\linewidth}
        \centering
        \begin{tikzpicture}
            \begin{axis}[
                width=\linewidth,
                height=4cm,
                grid=both,
                xlabel={},
                ylabel={},
                xmin = -1.5,
                xmax = 1.5,
                ymin = -0.6,
                ymax = 0.85,
                legend columns = 2,
                legend style={at={(0, 1)}, fill opacity = 0.7,anchor=north west, nodes={scale=0.66}},
                tick label style = {font=\tiny}
            ]
            \addplot+[only marks, mark=*, mark options={draw=red, fill=red}, mark size=1.5pt]
            coordinates {
                (-0.96218932, -0.10454969)
                (-1.08261271, -0.48829348)
                (-0.64005852,  0.22883317)
                (-1.06508457,  0.15476132)
                (-0.94375787, -0.11076457)
                (-0.80448651, -0.06211131)
                (-1.06576478, -0.15842935)
                (-0.90900839, -0.01983961)
                (-0.89094226, -0.12143714)
                (-0.97463443, -0.17845481)
                (-0.83170701,  0.03760702)
                (-0.9338858,   0.08210078)
            };
            \addplot+[only marks, mark=*, mark options={draw=blue, fill=blue}, mark size=1.5pt]
            coordinates {
                (0.7978485,   0.1566362)
                (1.41134056, -0.3276885)
                (0.65411771, -0.30096628)
                (1.16829178,  0.02574313)
                (1.21566849,  0.14448617)
                (1.04211436,  0.05680763)
                (0.9660479,   0.17369204)
                (0.77405681, -0.08437177)
                (1.04858777,  0.36028417)
                (0.84710718, -0.21581209)
                (0.88734256,  0.19385444)
                (0.9529989,   0.2648694)
            };
            \addplot+[only marks, mark=*, mark options={draw=black, fill=black}, mark size=1.5pt]
            coordinates {
                (-0.5, 0.0)
                (-0.3, 0.0)
                (-0.1, 0.0)
                ( 0.1, 0.0)
                ( 0.3, 0.0)
                ( 0.5, 0.0)
            };
            \legend{Cluster 1, Cluster 2, Bridge}
            \end{axis}
        \end{tikzpicture}
        \caption{Synthetic Dataset 2}
        \label{fig: bridgeprob}
    \end{minipage}
    \label{fig:two-plots}
\end{figure}
We compare the performance of our proposed algorithms SEQ-SPEC, IA-SEQ-SPEC, SEQ-SPEC-B, and IA-SEQ-SPEC-B, with the Spectral clustering algorithm \cite{ng2001spectral} (SPEC), and the sequential versions of Single Linkage clustering \cite{singh2025exponentially} (SEQ-SLINK) and $K$-Medoids clustering \cite{sreenivasan2023nonparametric} (SEQ-KMED).

\textbf{Synthetic Dataset 1:} 
We consider $30$ multivariate Gaussian distributed data sequences in $2$ clusters with mean vectors as shown in Fig.~\ref{fig: circleprob} and covariance matrix $0.4I_2$. For this simulation, we fix $q=0.7$, $p=4$, and $R=50$ in IA-SEQ-SPEC, and we represent it by IA-SEQ-SPEC-R=50 in Fig.~\ref{fig:circle}. SEQ-SPEC-R=50 in Fig.~\ref{fig:circle} corresponds to the SEQ-SPEC but finds the estimated configuration and checks the stopping condition only once every $R=50$ samples. From Fig.~\ref{fig:circle}, we observe that: (1) the proposed algorithms, SEQ-SPEC and IA-SEQ-SPEC, perform better than SPEC, SEQ-SLINK and SEQ-KMED, (2) IA-SEQ-SPEC-R=50 shows improvement over SEQ-SPEC-R=50, and it performs as well as SEQ-SPEC, (3) SEQ-KMED is not able to find the clusters in this example. Table \ref{Table: complexity} compares the computations required for IA-SEQ-SPEC and SEQ-SPEC.
\begin{table}[h!]
\centering
\begin{tabular}{|c|c|c|c|}
\hline
$M$ & 30 (10+20) & 45 (15+30) & 60 (20+40) \\ \hline
IA-SEQ-SPEC & 2098 & 5450 & 12169 \\ \hline
SEQ-SPEC & 27000 & 91125 & 216000 \\ \hline
\end{tabular}
\caption{Average number of computations for eigen decomposition for Synthetic Dataset 1 for different values of $M$}
\label{Table: complexity}
\end{table}

\textbf{Synthetic Dataset 2:} 
We consider multivariate Gaussian-distributed data sequences in $2$ clusters, with mean vectors as shown in Fig.~\ref{fig: bridgeprob}. Each cluster comprises 12 data sequences, connected by a bridge comprising 6 data sequences.  The algorithms are allowed to group the data sequences in the bridge into either of the clusters. This problem of clusters/blobs connected by a bridge is well-studied  \cite{scripps2006clustering, laplante2024spectral}, and has many applications in community detection \cite{granovetter1973strength}, image segmentation \cite{shi2000normalized}, etc. In Fig.~\ref{fig:bridge}, we can observe that the IA-SEQ-SPEC performs as well as SEQ-SPEC and SEQ-KMED in this example. For such problems, the SLINK clustering is not suitable because it tends to link two clusters through the data sequences in the bridge.

\textbf{Synthetic Dataset 3:}
We consider $40$ multivariate Gaussian distributed data sequences in $3$ clusters with mean vectors as shown in Fig.~\ref{fig:circleprob3clus} and covariance matrix $0.4I_2$. Fig.~\ref{fig:circle3clus} shows the performance of the various algorithms. For this simulation, we fix $B=30$ for the memory-efficient algorithms. We observe that the efficient algorithms IA-SEQ-SPEC, SEQ-SPEC-B, and IA-SEQ-SPEC-B perform comparably to SEQ-SPEC, and all of them outperform SPEC.

\textbf{Real-world Dataset 1:}
We consider the MovieLens dataset, which comprises user ratings for various movies and genres \cite{cantador2011second}. We consider $M=6$ users and group them into $K=4$ clusters. For clustering, we observe users’ ratings for the Comedy and Drama genres; hence, at each round, the algorithm observes $d=2$ dimensional samples. Fig.~\ref{fig:movielens} shows the performance of the various algorithms. We observe that SEQ-SPEC and SEQ-KMED show the best performance, followed by IA-SEQ-SPEC and then SEQ-SLINK.

\textbf{Real world Dataset 2:} 
We consider the Modified National Institute of Standards and Technology (MNIST) dataset \cite{lecun2002gradient}, which consists of images of digits from 0 to 9. 
We consider each digit as a cluster. We divide the set of all data points corresponding to digit $i$ into two subsets, where each of these 2 subsets is a data sequence corresponding to the $i^{th}$ digit. Hence, we have $20$ data sequences and $10$ clusters. Fig.~\ref{fig:MNIST} shows that our proposed algorithms, SEQ-SPEC, and IA-SEQ-SPEC perform better than SEQ-KMED and SEQ-SLINK.
\vspace{-0.3em}

\section{Conclusions} \label{sec: conclusion}
We proposed sequential spectral clustering algorithms: SEQ-SPEC, a lower-complexity approximate IA-SEQ-SPEC, and their memory-efficient versions, SEQ-SPEC-B and IA-SEQ-SPEC-B, for clustering data sequences. We showed that the SEQ-SPEC algorithm stops in finite time almost surely and is exponentially consistent. Through simulations on both synthetic and real-world datasets, we showed that both SEQ-SPEC and IA-SEQ-SPEC perform better than SPEC, SEQ-SLINK, and SEQ-KMED. In simulations, we observe that the IA-SEQ-SPEC and memory-efficient versions perform comparably to SEQ-SPEC. The proposed algorithms are able to work well for different types of cluster configurations, while SEQ-SLINK and SEQ-KMED fail for some configurations.

\balance
\bibliographystyle{IEEEtran}
\bibliography{ref}



\newpage
\onecolumn
\appendices
\begin{center}
    \Large \bfseries Technical Appendix
\end{center}
\section{Proof of Exponential consistency of FSS SPEC (Theorem 1)}
\begin{proof}
    Let us define the event $F \coloneqq \left\{ \forall i, j \in [n], i<j, \left| \hat{d}_{ij}(t) - d_{ij} \right|<\epsilon \right\}$. Assume that the event $F$ holds. Consider $\epsilon < \frac{m}{2c}$, where $m$ and $c$ is same as defined in the beginning of Section \ref{sec: techlemcla}.
    Now we have the results in Claims \ref{claim: delta assump}, \ref{claim: epsilon1 assump}, \ref{claim: epsilon2 assump}, and \ref{claim: C assump}. Now we invoke Lemma \ref{lemma: ngTheorem2}. 
    
    First, we derive the sufficient condition to apply Lemma \ref{lemma: ngTheorem2}. 
    \begin{equation}
        \delta - A_5\epsilon > \boldsymbol{\epsilon^{'}}(2+\sqrt{2}), 
    \end{equation}
    where $\boldsymbol{\epsilon^{'}} = \sqrt{K(K-1)(\epsilon_1+A_6\epsilon)+K(\epsilon_2+A_7\epsilon)^2}$ with $A_5, A_6, A_7$ is as defined in Claims \ref{claim: delta assump}, \ref{claim: epsilon1 assump}, \ref{claim: epsilon2 assump} respectively.
    Since we consider $\epsilon<\frac{m}{2c}$, we set $\epsilon_T = \frac{m}{2c}$ in Claim \ref{claim: boldempsilon bound} and by using Claim \ref{claim: boldempsilon bound}, we can say that it sufficient to ensure the following condition. 
    \begin{equation}
        \delta - A_5\epsilon > \left[ \boldsymbol{\epsilon} + \Tilde{C}\left(\frac{m}{2c}\right)\epsilon \right](2+\sqrt{2}),
    \end{equation}
    where $\boldsymbol{\epsilon} = \sqrt{K(K-1)\epsilon_1+K\epsilon_2}$ and $\Tilde{C}(\cdot)$ is as defined in Claim \ref{claim: boldempsilon bound}.
    This on simplification, we get the sufficient condition as follows.
    \begin{equation}
        \epsilon < \frac{\delta - (2+\sqrt{2})\boldsymbol{\epsilon}}{C_3} \ \ \text{where, } C_3 = \Tilde{C}\left( \frac{m}{2c} \right)(2+\sqrt{2}) + A_5.
    \end{equation}
    Here, we assume $\delta>(2+\sqrt{2})\boldsymbol{\epsilon}$. Hence, to use Lemma \ref{lemma: ngTheorem2}, $\epsilon$ should satisfy the following condition.
    \begin{equation}
        \epsilon < w_0 \ \ \text{where, } w_0 = \min\left[ \frac{\delta - (2+\sqrt{2})\boldsymbol{\epsilon}}{C_3}, \frac{m}{2c} \right].
    \end{equation}
     
    Let's assume, for all $k \in [K]$, for all $i \in D_k$, $\left\| \hat{Y}_i(t)-r_k \right\|<r$ for some $r>0$. For the K-Means algorithm on the set of $ K$-dimensional data points $\left\{\hat{Y}_i(t), i \in [M]\right\}$ to give the correct clustering, it should satisfy the condition that the maximum intra-cluster distance is less than the minimum inter-cluster distance. Since $r_k$'s are orthonormal, the distance between any two $r_k$'s is $\sqrt{2}$. Also, estimated spectral points corresponding to the data sequences whose true cluster is $k$, lie inside the ball of radius $r$ around the orthonormal vector $r_k$. Hence, to get correct clustering, it is enough to satisfy the condition $2r<\sqrt{2}-2r$, where $2r$ is the maximum intra-cluster distance and $\sqrt{2}-2r$ is the minimum inter-cluster distance. Hence, we choose $r<0.35$.
    Therefore $\left\| \hat{Y}_i(t)-r_k \right\|<r$ for all $i \in [M]$ ensures the correct clustering. 

    Hence it is sufficient to satisfy the condition that $\frac{1}{M} \sum_{k \in [K]} \sum_{i \in D_k}\left\| \hat{Y}_i-r_k \right\|< \frac{r}{M}$. Now on using Lemma \ref{lemma: ngTheorem2}, we can say that to get correct clusters, it is sufficient to satisfy the following condition.
    \begin{equation}
        4(C+A_8\epsilon)(4+2\sqrt{K})^2\frac{\boldsymbol{(\epsilon^{'}})^2}{((\delta-A_5\epsilon)-\sqrt{2}\boldsymbol{\epsilon^{'}})^2} < \frac{r}{M},
    \end{equation}
    where $A_8$ is defined as in Claim \ref{claim: C assump}. We also have that $\delta-A_5\epsilon>(2+\sqrt{2})\boldsymbol{\epsilon^{'}}$. Therefore we have $(\delta-A_5\epsilon)-\sqrt{2}\boldsymbol{\epsilon^{'}}>2\boldsymbol{\epsilon^{'}}$. Hence, we get the sufficient condition as follows.
    \begin{equation}
    \begin{aligned}
        (C+A_8\epsilon)(4+2\sqrt{K})^2 &< \frac{r}{M} \\
        \left[C + A_8\epsilon\right] &< \frac{r}{M(4+2\sqrt{K})^2} \\
        \epsilon &< \frac{1}{A_8} \left[ \frac{r}{M(4+2\sqrt{K})^2} - C \right] =w_1 \text{(say)}.
        \end{aligned}
    \end{equation}
    We assume that $MC(4+2\sqrt{K})^2<r$. Hence, the clustering algorithm to give the correct output, it should satisfy the condition that $\epsilon< \epsilon_0$, where $\epsilon_0 \coloneqq \min\left\{w_0, w_1\right\}$. Formally, we can say that the algorithm gives the correct clustering if the following event holds. 
    \begin{equation}
        \left\{ \forall i, j \in [n], i<j \left| \hat{d}_{ij}(t) - d_{ij} \right| < \epsilon_0 \right\}.
    \end{equation}
    Hence, the probability of error at time $t$, which we denote by $\mathbb{P}[E_t]$ can be upper bounded as follows.
    \begin{equation}
        \mathbb{P}[E_t] \leq \mathbb{P}\left[ \exists i, j \in [M], i<j, \left| \hat{d}_{ij}(t) - d_{ij} \right| > \epsilon_0 \right]
    \end{equation}
    Let's define $t_0 \coloneqq \frac{\alpha_2}{\epsilon_0^2}$, where $\alpha_2$ is as defined in Claim \ref{claim: mmd conc bound}. Now, from Claim \ref{claim: mmd conc bound}, for all $t\geq t_0$, we have,
    \begin{equation}
        \mathbb{P}[E_t] \leq M^2\exp{\left( \frac{-\epsilon_0^2t}{16B} \right)}.
    \end{equation}
    Hence proved.
\end{proof}

\section{Proof that SEQ-SPEC stops in finite time (Theorem 2)}

\begin{proof}
    Now, we upper bound the probability that the stopping time is greater than some $t \in \mathbb{N}$ as follows.
    \begin{align}
        \mathbb{P}[N>t] &= \mathbb{P}\left[\{N>t\}\cap E_t\right] + \mathbb{P}\left[\{N>t\}\cap E_t^C\right] \\
        &\leq \mathbb{P}\left[E_t\right] + \mathbb{P}\left[\{N>t\}\cap E_t^C\right]
    \end{align}
    We define $T_t \coloneqq \arcsin\left(\frac{C}{\sqrt{t}}\right)$. The event that the stopping time is greater than $t$ is the same as the event that the stopping condition of the SEQ-SPEC is not met till time $t$. Hence, we further upper bound as follows.
    \begin{align}
        &= \mathbb{P}\left[E_t\right] + \mathbb{P}\left[\{\Gamma_{t_1}<T_{t_1}, \forall t_1\leq t\}\cap E_t^C\right] \\
        &\leq \mathbb{P}\left[E_t\right] + \mathbb{P}\left[\{\Gamma_{t}<T_{t}\}\cap E_t^C\right] \\
        &= \mathbb{P}\left[E_t\right] + \mathbb{P}\left[\left\{\min_{k\neq l}\min_{i \in \hat{D}_k(t)}\min_{j \in \hat{D}_l(t)}\left\| \hat{Y}_i(t) - \hat{Y}_j(t) \right\|<T_{t}\right\}\cap E_t^C\right]         
    \end{align}
    Under the event that the estimated clusters are correct at time $t$, that is $E_t^C$, we can write $\hat{D}_k(t) = D_k$ for all $k \in [K]$. Therefore, we have the following.
    \begin{align}
        &\leq \mathbb{P}\left[E_t\right] + \mathbb{P}\left[\left\{\min_{k\neq l}\min_{i \in D_k}\min_{j \in D_l}\left\| \hat{Y}_i(t) - \hat{Y}_j(t) \right\|<T_{t}\right\}\right] \\
        &\leq \mathbb{P}\left[E_t\right] + \mathbb{P}\left[\bigcup_{k\neq l}\bigcup_{i \in D_k}\bigcup_{j \in D_l}\left\{\left\| \hat{Y}_i(t) - \hat{Y}_j(t) \right\|<T_{t}\right\}\right] \\
        &\leq \mathbb{P}\left[E_t\right] + \sum_{k\neq l}\sum_{i \in D_k}\sum_{j \in D_l}\mathbb{P}\left[\left\{\left\| \hat{Y}_i(t) - \hat{Y}_j(t) \right\|<T_{t}\right\}\right] 
    \end{align}
    From Lemma \ref{theorem: expconsfss}, there exist $t_0\in \mathbb{N}$ such that for all $t\geq t_0$, we have $\mathbb{P}\left[E_t\right] \leq M^2\exp{\left( -\alpha_0 t \right)}$ for $\alpha_0>0$. Hence, we have for all $t\geq t_0$, 
    \begin{equation} \label{Th1: eq 1}
        \mathbb{P}[N>t] \leq M^2\exp{\left( -\alpha_0 t \right)} + \sum_{k\neq l}\sum_{i \in D_k}\sum_{j \in D_l}\mathbb{P}\left[\left\{\left\| \hat{Y}_i(t) - \hat{Y}_j(t) \right\|<T_{t}\right\}\right] 
    \end{equation}
    We say $d_H$ as the true minimum inter-cluster distance in the spectral domain, and it is defined as $d_H \coloneqq \min_{k\neq l} \min_{i \in D_k} \min_{j \in D_l} \left\| Y_i - Y_j \right\|$. So, we have $\left\| Y_i - Y_j \right\|\geq d_H$ for all $i$ and $j$ from different clusters. Hence we can upper bound the term  $\mathbb{P}\left[\left\{\left\| \hat{Y}_i(t) - \hat{Y}_j(t) \right\|<T_{t}\right\}\right]$ as follows.
    \begin{align}
        \mathbb{P}\left[ \left\| \hat{Y}_i(t) - \hat{Y}_j(t) \right\| < T_t \right] &= \mathbb{P}\left[ \left\| Y_i - Y_j \right\| - \left\| \hat{Y}_i(t) - \hat{Y}_j(t) \right\| > \left\| Y_i - Y_j \right\| - T_t \right] \\
        &\leq \mathbb{P}\left[ \left|\left\| Y_i - Y_j \right\| - \left\| \hat{Y}_i(t) - \hat{Y}_j(t) \right\|\right| > d_H - T_t \right].
    \end{align}
    We have $T_t = \arcsin\left( \frac{C}{\sqrt{t}} \right)$. We define $t_1$ such that $T_{t_1} = \frac{d_H}{2}$, that is $t_1 \coloneqq \left[ \frac{C}{\sin{d_H}} \right]^2$. For all $t>t_1$, we have $d_H-T_t > \frac{d_H}{2}$. Therefore, we can say $d_H-T_t > \delta d_H$ for some $\delta \in \left( 0, \frac{1}{2} \right)$. Hence for all $t>\max\{t_0, t_1\}$, we can upper bound as follows.
    \begin{equation}
        \mathbb{P}\left[ \left\| \hat{Y}_i(t) - \hat{Y}_j(t) \right\| < T_t \right] \leq \mathbb{P}\left[ \left|\left\| Y_i - Y_j \right\| - \left\| \hat{Y}_i(t) - \hat{Y}_j(t) \right\|\right| > \delta d_H \right].
    \end{equation}
    It can be geometrically verified that $\left|\left\| Y_i - Y_j \right\| - \left\| \hat{Y}_i(t) - \hat{Y}_j(t) \right\|\right| \leq \left\| \hat{Y}_i(t) - Y_i \right\| + \left\|\hat{Y}_j(t) - Y_j\right\|$. Now, for all $t > \max\{t_0, t_1\}, $we get, 
    \begin{align}
        \mathbb{P}\left[ \left\| \hat{Y}_i(t) - \hat{Y}_j(t) \right\| < T_t \right] &\leq \mathbb{P}\left[ \left\|\hat{Y}_i(t) - Y_i\right\| +  \left\|\hat{Y}_j(t) - Y_j\right\| > \delta d_H \right] \\
        &\leq \mathbb{P}\left[ \left\{\left\|\hat{Y}_i(t) - Y_i\right\| > \frac{\delta d_H}{2} \right\} \bigcup \left\{\left\|\hat{Y}_j(t) - Y_j\right\| > \frac{\delta d_H}{2} \right\} \right] \\
    \end{align}
    Now to use Lemma \ref{lemma: specpt bound}, we choose $\delta \in \left( 0, \frac{1}{2} \right)$ small enough such that $\frac{\delta d_H}{2} <  A_1\arcsin\left( \alpha_1 \epsilon_3\right)$ where $A_1, \alpha_1, \epsilon_3$ are as defined in Lemma \ref{lemma: specpt bound}. Now, by using Lemma \ref{lemma: specpt bound}, we can upper bound as follows for all $t \geq \max\{t_0, t_1\}$.
    \begin{align}
        \mathbb{P}\left[ \left\| \hat{Y}_i(t) - \hat{Y}_j(t) \right\| < T_t \right] &\leq \mathbb{P}\left[ \exists i, j \in [n], i<j, \left| \hat{d}_{ij}(t) - d_{ij} \right| > \eta \right], \ \ \text{where, } \eta= \frac{1}{\alpha_1}\sin\left( \frac{\delta d_H }{2A_1} \right)
    \end{align}
    Now, consider $t_2$ as defined in Claim \ref{claim: mmd conc bound}. On using Claim \ref{claim: mmd conc bound}, for all $t\geq \{t_0, t_1, t_2\}$, we get, 
    \begin{align}
        \mathbb{P}\left[ \left\| \hat{Y}_i(t) - \hat{Y}_j(t) \right\| < T_t \right] &\leq M^2 \exp{\left( \frac{-\eta^2t}{16B} \right)}.
    \end{align}
    Now on using the above derived upper bound for $\mathbb{P}\left[ \left\| \hat{Y}_i(t) - \hat{Y}_j(t) \right\| < T_t \right]$ on equation \eqref{Th1: eq 1}, we get for all $t\geq \max\left\{ t_0, t_1, t_2 \right\}$, 
    \begin{align}
        \mathbb{P}\left[N>t\right] &\leq M^2\exp{\left( -\alpha_0 t \right)} + \sum_{k\neq l}\sum_{i \in D_k}\sum_{j \in D_l} M^2 \exp{\left( \frac{-\eta^2t}{16B} \right)} \\
        &\leq M^2\exp{\left( -\alpha_0 t \right)} + K^2M^4 \exp{\left( \frac{-\eta^2t}{16B} \right)}. \label{th1: eq: 3}
    \end{align}
    Therefore from the above derived upper bound for $\mathbb{P}\left[N>t\right]$, we can say $\lim_{t\rightarrow \infty} \mathbb{P}[N>t] = 0$. Hence proved. 
\end{proof}

\section{Proof of Exponential consistency of SEQ-SPEC (Theorem 3)}
To prove Theorem \ref{theorem: expconsSEQ}, first we show that SEQ-SPEC is universally consistent, i.e., the probability of error goes to 0 as $C$ tends to $\infty$ (Theorem \ref{theorem: universalconsis}). Then we show that $\frac{N}{C^2}$ converges to a problem dependent constant (Theorem \ref{unif}). We use Theorems \ref{theorem: universalconsis} and \ref{unif} to prove Theorem \ref{theorem: expconsSEQ}.

\begin{theorem} \label{theorem: universalconsis}
    The proposed sequential spectral clustering algorithm is universally consistent, that is, $\lim_{C\rightarrow\infty}\mathbb{P}[E_N] = 0$.
\end{theorem}
\begin{proof}
    We define $t_M \coloneqq \left(C+t_0\right)^2$, where $t_0$ is as defined in Lemma \ref{theorem: expconsfss}. Now, we can upper bound the error event $E_N$ as follows.
    \begin{align}
        \mathbb{P}[E_N] &= \sum_{t=1}^\infty \mathbb{P}[N=t, E_t] \\
        &= \sum_{t=1}^{t_M} \mathbb{P}[N=t, E_t] + \sum_{t>t_M}^{\infty} \mathbb{P}[N=t, E_t] \\
        &\leq \sum_{t=1}^{t_M} \mathbb{P}[N=t, E_t] + \sum_{t>t_M}^{\infty} \mathbb{P}[E_t].
    \end{align}
    Since $t_M > t_0$, we apply Lemma \ref{theorem: expconsfss} on the second term of the above expression. 
    \begin{align}
        &\leq \sum_{t=1}^{t_M} \mathbb{P}[N=t, E_t] + \sum_{t>t_M}^{\infty} n^2\exp{\left(-\alpha_0 t\right)}  \\
        &= \sum_{t=1}^{t_M} \mathbb{P}[N=t, E_t] + \frac{n^2 \exp{\left( -\alpha_0(t_M+1) \right)}}{1 - \exp{\left(-\alpha_0\right)}} \\
        &\leq \sum_{t=1}^{t_M} \mathbb{P}[N=t, E_t] + \frac{n^2 \exp{\left( -\alpha_0 C^2 \right)}}{1 - \exp{\left(-\alpha_0\right)}} \ \ \left(\because t_M > C^2\right). \label{th2: eq: 3}
    \end{align}
    Now we upper bound the first term of the above expression. The threshold in our stopping rule is $T_t = \arcsin\left( \frac{C}{\sqrt{t}} \right)$. Since the domain of $\arcsin{(\cdot)}$ is $[-1, 1]$, the algorithm will stop only for time $t\geq t_L$, where $t_L=C^2$. Hence $\mathbb{P}[N=t] = 0$ for all $t< t_L$. Therefore, we can write as follows.
    \begin{align}
        \sum_{t=1}^{t_M} \mathbb{P}[N=t, E_t] &= \sum_{t=t_L}^{t_M} \mathbb{P}[N=t, E_t]\\
        &= \sum_{t=t_L}^{t_M} \mathbb{P}\left[\left\{ \Gamma_s < T_s, \forall s<t \right\}\bigcap\left\{ \Gamma_t>T_t \right\}\bigcap E_t\right] \\
        &\leq \sum_{t=t_L}^{t_M} \mathbb{P}\left[\left\{ \Gamma_t>T_t \right\}\bigcap E_t\right] \\
        &= \sum_{t=t_L}^{t_M}\mathbb{P}\left[ \left\{ \min_{k\neq l} \min_{i_1 \in \hat{D}_k(t)} \min_{j_1 \in \hat{D}_l(t)} \left\| \hat{Y}_{i_1}(t) - \hat{Y}_{j_1}(t) \right\| > T_t \right\}\bigcap E_t \right] 
    \end{align}
    Under the error event at time $t$, $E_t$, we can say for some $k \in [K]$ there exist $i, j \in D_k$ such that $i \in \hat{D}_p$ and $j \in \hat{D}_q$ for some $p, q \in [K]$ with $p\neq q$. Also, we have $\left\| \hat{Y}_i(t) - \hat{Y}_j(t) \right\| \geq \min_{k\neq l} \min_{i_1 \in \hat{D}_k(t)} \min_{j_1 \in \hat{D}_l(t)} \left\| Y_{i_1}(t) - Y_{j_1}(t) \right\|$. Therefore, we can further upper bound as follows.
    \begin{align}
        &\leq \sum_{t=t_L}^{t_M}\mathbb{P}\left[ \left\| \hat{Y}_{i}(t) - \hat{Y}_{j}(t) \right\| > T_t \right]  \\
        &= \sum_{t=t_L}^{t_M}\mathbb{P}\left[ \left\| \hat{Y}_{i}(t) - \hat{Y}_{j}(t) \right\| - \left\| Y_{i} - Y_{j} \right\| > T_t - \left\| Y_{i} - Y_{j} \right\| \right] \\
        &= \sum_{t=t_L}^{t_M}\mathbb{P}\left[ \left|\left\| \hat{Y}_{i}(t) - \hat{Y}_{j}(t) \right\| - \left\| Y_{i} - Y_{j} \right\|\right| > T_t - \left\| Y_{i} - Y_{j} \right\| \right]
    \end{align}
    We use $d_L$ to denote the maximum intra cluster distance and is formally defined as $d_L\coloneqq \max_{k \in [K]} \max_{i, j \in D_k, i \neq j}\left\| Y_{i} - Y_{j} \right\|$. Since $i, j$ are from the same clusters, we can say $\left\| Y_{i} - Y_{j} \right\| \leq d_H$. Hence we can further upper bound as follow.
    \begin{equation} \label{th2: eq: 1}
        \sum_{t=1}^{t_M} \mathbb{P}[N=t, E_t] \leq \sum_{t=t_L}^{t_M}\mathbb{P}\left[\left|\left\| \hat{Y}_{i}(t) - \hat{Y}_{j}(t) \right\| - \left\| Y_{i} - Y_{j} \right\|\right| > T_t - d_L \right].
    \end{equation}
    
    Let us assume $T_t-d_L > \delta d_L$ for some $\delta \in (0, 1)$ and $t\leq t_M$. Now we derive a condition on $C$ for this assumption to be satisfied.
    \begin{align}
        T_t-d_L &> \delta d_L \\
        (1-\delta)\arcsin\left( \frac{C}{\sqrt{t}} \right) &> d_L \\ 
        C &> \sqrt{t}\sin\left( \frac{d_L}{(1-\delta)} \right).
    \end{align}
    Since $t\leq t_M$ in equation \eqref{th2: eq: 1}, it is sufficient to satisfy the following condition.
    \begin{align}
        C &> \sqrt{t_M}\sin\left( \frac{d_L}{(1-\delta)} \right) \\
        C &> (C+t_0)\sin\left( \frac{d_L}{(1-\delta)} \right) \ \ (\because t_M = (t_0+C)^2)\\
        C\left(1-\sin\left( \frac{d_L}{(1-\delta)} \right) \right) &> t_0 \sin\left( \frac{d_L}{(1-\delta)} \right) \\
        C &> C_1 \ \ \text{where, } C_1 = \frac{t_0 \sin\left( \frac{d_L}{(1-\delta)} \right)}{1-\sin\left( \frac{d_L}{(1-\delta)} \right)}.
    \end{align}
    Therefore, the equation \eqref{th2: eq: 1} can be further upper bounded as, 
    \begin{align}
        \sum_{t=1}^{t_M} \mathbb{P}[N=t, E_t] &\leq \sum_{t=t_L}^{t_M}\mathbb{P}\left[\left|\left\| \hat{Y}_{i}(t) - \hat{Y}_{j}(t) \right\| - \left\| Y_{i} - Y_{j} \right\|\right| > \delta T_t \right] \\
        &\leq \sum_{t=t_L}^{t_M}\mathbb{P}\left[\left|\left\| \hat{Y}_{i}(t) - \hat{Y}_{j}(t) \right\| - \left\| Y_{i} - Y_{j} \right\|\right| > \delta T_{t_M} \right] \ \ \left(\because T_t \geq T_{t_M} \forall t\leq t_M \right)
    \end{align}
    Let us define $\delta^{'}$ such that $\delta T_{t_M} = 2A_1\arcsin\left( \frac{\delta^{'}C}{\sqrt{t_M}} \right)$, that is $\delta^{'} =  \frac{\sqrt{t_M}}{C}\sin\left( \frac{\delta}{2A_1} \arcsin\left( \frac{C}{\sqrt{t_M}} \right) \right)$, for some $A_1>0$.
    Now we can upper bound as follows.
    \begin{equation} \label{th2: eq: 2}
        \sum_{t=1}^{t_M} \mathbb{P}[N=t, E_t] \leq \sum_{t=t_L}^{t_M}\mathbb{P}\left[ \left|\left\| \hat{Y}_{i}(t) - \hat{Y}_{j}(t) \right\| - \left\| Y_{i} - Y_{j} \right\|\right| > 2A_1\arcsin\left( \frac{\delta^{'}C}{\sqrt{t_M}} \right) \right].
    \end{equation}
    Now we will derive a lower bound for $\delta^{'}$ which independent of $C$ and depends only on $\delta$.
    \begin{align}
        \delta^{'} &= \frac{\sqrt{t_M}}{C}\sin\left( \delta \arcsin\left( \frac{C}{\sqrt{t_M}} \right) \right) \\
        &= \frac{C+t_0}{C}\sin\left( \delta \arcsin\left( \frac{C}{C+t_0} \right) \right) \hspace{2cm} (\because t_M = (t_0+C)^2) \\
        &> \sin\left( \delta \arcsin\left( \frac{1}{1+\frac{t_0}{C}} \right) \right) \hspace{2cm} \left( \because \frac{C+t_0}{C}>1 \right)\\
    \end{align}
    Let us define $C_2 = t_0$. Therefore, for $C>\max\{C_1, C_2\}$ we further lower bound the above expression as follows.
    \begin{align}
        \delta^{'} &> \sin\left( \delta \arcsin\left( \frac{1}{2} \right) \right) \\
        &= \sin\left(\delta \frac{\pi}{6}\right) = \delta^{''} \text{(say)}.
    \end{align}
    Since $\arcsin(\cdot)$ is a monotonically increasing function, for all $C>\max\{C_1, C_2\}$, we can upper bound the equation \eqref{th2: eq: 2} as follows.
    \begin{equation}
        \sum_{t=1}^{t_M} \mathbb{P}[N=t, E_t] \leq \sum_{t=t_L}^{t_M}\mathbb{P}\left[ \left|\left\| \hat{Y}_{i}(t) - \hat{Y}_{j}(t) \right\| - \left\| Y_{i} - Y_{j} \right\|\right| > 2A_1\arcsin\left( \frac{\delta^{''}C}{\sqrt{t_M}} \right) \right].
    \end{equation}
    Let us define $C^{'} = \frac{C}{\alpha_1}$ for some $\alpha_1>0$. We have the following.
    \begin{align}
        \sum_{t=1}^{t_M} \mathbb{P}[N=t, E_t] &\leq \sum_{t=t_L}^{t_M}\mathbb{P}\left[ \left|\left\| \hat{Y}_{i}(t) - \hat{Y}_{j}(t) \right\| - \left\| Y_{i} - Y_{j} \right\|\right| > 2A_1\arcsin\left( \alpha_1\frac{\delta^{''}C^{'}}{\sqrt{t_M}} \right) \right] 
    \end{align}
    It can be geometrically verified that $\left|\left\| Y_i - Y_j \right\| - \left\| \hat{Y}_i(t) - \hat{Y}_j(t) \right\|\right| \leq \left\| \hat{Y}_i(t) - Y_i \right\| + \left\|\hat{Y}_j(t) - Y_j\right\|$. Hence we can further upper bound the above expression as follows.
    \begin{align}
        &\leq \sum_{t=t_L}^{t_M}\mathbb{P}\left[ \left\| \hat{Y}_i(t) - Y_i \right\| + \left\|\hat{Y}_j(t) - Y_j\right\| > 2A_1\arcsin\left( \alpha_1\frac{\delta^{''}C^{'}}{\sqrt{t_M}} \right) \right] \\
        &\leq \sum_{t=t_L}^{t_M}\mathbb{P}\left[ \left\{\left\| \hat{Y}_i(t) - Y_i \right\| > A_1\arcsin\left( \alpha_1\frac{\delta^{''}C^{'}}{\sqrt{t_M}} \right) \right\} \bigcup \left\{\left\|\hat{Y}_j(t) - Y_j\right\| > A_1\arcsin\left( \alpha_1\frac{\delta^{''}C^{'}}{\sqrt{t_M}} \right) \right\} \right]
    \end{align}
    
    From the definition of $C^{'}$ and $t_M$, we can observe that $\frac{C^{'}}{\sqrt{t_M}}$ can be upper bounded by a constant $\frac{1}{\alpha_1}$. So, we can choose the value of $\delta \in (0, 1)$ small enough independent of $C$ such that $\frac{\delta^{''}C^{'}}{\sqrt{t_M}} < \epsilon_3$, where $\epsilon_3$ is as defined in Lemma \ref{lemma: specpt bound}. Now we choose the value of $A_1$ and $\alpha_1$ as defined in Lemma \ref{lemma: specpt bound} and use Lemma \ref{lemma: specpt bound} in the above equation to get the following upper bound.
    \begin{align}
        \sum_{t=1}^{t_M} \mathbb{P}[N=t, E_t] &\leq \sum_{t=t_L}^{t_M}\mathbb{P}\left[ \exists i, j \in [K], i<j, \left| d_{ij}(t) - d_{ij}(t) \right| > \frac{\delta^{''}C^{'}}{\sqrt{t_M}} \right] 
    \end{align}
    To use Claim \ref{claim: mmd conc bound} in the above expression, we need to ensure that $t_L \geq t_2\left( \frac{\delta^{''}C^{'}}{\sqrt{t_M}} \right)$, where $t_2(\epsilon) = \frac{\alpha_2}{\epsilon^2}$ as defined in Claim \ref{claim: mmd conc bound}. Since $t_L=C^2$, it is equivalent to ensure that $C\geq \sqrt{\frac{\alpha_2t_M}{\left(\delta^{'''}\right)^2\left(C^{'}\right)^2}}$. By using the facts $t_M = (C+t_0)^2$, $C>C_2$ and $C^{'} = \frac{C}{\alpha_1}$, it can be verified that $\frac{\sqrt{t_M}}{C^{'}}\leq 2\alpha_1$. Hence it is sufficient to ensure $C \geq \sqrt{\frac{4\alpha_1^2\alpha_2}{\left(\delta^{''}\right)^2}}$. Let us define $C_3 = \sqrt{\frac{4\alpha_1^2\alpha_2}{\left(\delta^{''}\right)^2}}$. Hence for $C\geq \max\{C_1, C_2, C_3\}$, we can upper bound the above expression by using Claim \ref{claim: mmd conc bound} as follows.
    \begin{align}
        \sum_{t=1}^{t_M} \mathbb{P}[N=t, E_t] &\leq \sum_{t=t_L}^{t_M}M^2 \exp{\left( \frac{-(\delta^{''})^2(C^{'})^2t}{16Bt_M} \right)} \\
        &\leq \left(t_M - t_L +1\right) M^2 \exp{\left( \frac{-(\delta^{''})^2(C^{'})^2t_L}{16Bt_M} \right)}
    \end{align}
    
    Since $C>C_2$, we have $\frac{t_L}{t_M} = \frac{C^2}{(C+t_0)^2} > \frac{1}{4}$. Hence we get, 
    \begin{align}
        \sum_{t=1}^{t_M} \mathbb{P}[N=t, E_t] &\leq \left(t_M - t_L +1\right) M^2 \exp{\left( \frac{-(\delta^{''})^2(C^{'})^2}{64B} \right)} \\
        &= \left( 2Ct_0 + t_0^2 + 1 \right) M^2 \exp{\left( \frac{-(\delta^{''})^2C^2}{64B\alpha_1^2} \right)}.
    \end{align}
    Hence, on using the above derived upper bound on the equation \eqref{th2: eq: 3}, for all $C>\max\{C_1, C_2, C_3\}$, probability of error is bounded as follows.
    \begin{equation} \label{th2: eq: 4}
        \mathbb{P}[E_N] \leq \left( 2Ct_0 + t_0^2 + 1 \right) M^2 \exp{\left( \frac{-(\delta^{''})^2C^2}{64B\alpha_1^2} \right)} + \frac{M^2 \exp{\left( -\alpha_0 C^2 \right)}}{1 - \exp{\left(-\alpha_0\right)}}
    \end{equation}
    Therefore we get $\lim_{C\rightarrow\infty}\mathbb{P}[E_N] = 0$. Hence proved.
\end{proof}

\begin{theorem} \label{unif}
    The stopping time $N$ of the sequential spectral clustering algorithm satisfies $\lim_{C \rightarrow \infty} \mathbb{E}\left[ \left| \frac{N}{C^2} - \frac{1}{\sin^2\left( d_H \right)} \right| \right] = 0$.
\end{theorem}
\begin{proof}
    From the definition of the finite stopping time and from the Theorem \ref{theorem: finitestoptime} which presents that the the proposed algorithm stops in finite time, we have the following.
    \begin{equation}
        \mathbb{P}\left[ \Gamma_N > \arcsin\left( \frac{C}{\sqrt{N}} \right) \right] = 1 \text{ and } \mathbb{P}\left[ \Gamma_{N-1} \leq \arcsin\left( \frac{C}{\sqrt{N-1}} \right) \right] = 1.
    \end{equation}
    Hence with probability $1$, we have the following inequality.
    \begin{equation} \label{th3: eq: 1}
        \frac{1}{\sin^2\left( \Gamma_N \right)} < \frac{N}{C^2} \leq \frac{1}{\sin^2\left( \Gamma_{N-1} \right)} +  \frac{1}{C^2}.
    \end{equation}
    Since $\sin(\cdot)\leq 1$, we can say $1\leq \frac{1}{\sin^2\left( \Gamma_N \right)}$. Therefore, from the equation \eqref{th3: eq: 1}, we can say $N\geq C^2$. Hence the stopping time of the algorithm $N$ tends to $\infty$ as $C$ tends to $\infty$. 
    
    From the properties of the MMD distance, estimated MMD distance $\hat{d}_{ij}(t)$ converges almost surely to the true distance $d_{ij}$. More formally, for all $\nu_1>0$, there exists $t_4\in \mathbb{N}$ such that for all $t>t_4$, we have $\left| \hat{d}_{ij}(t) - d_{ij} \right|<\nu_1$ for all $i, j \in \mathbb{N}$. Consider the values of $\nu_1$ which are small enough to use Lemma \ref{lemma: specpt bound}. Now from Lemma \ref{lemma: specpt bound}, we have $\left| \hat{Y}_j(t) - Y_j \right|< A_1\arcsin\left( \alpha_1\nu_1 \right)$ for all $j \in [M]$, where $A_1, \alpha_1$ are defined as in Lemma \ref{lemma: specpt bound}. It can be geometrically verified that $\left|\left\| Y_i - Y_j \right\| - \left\| \hat{Y}_i(t) - \hat{Y}_j(t) \right\|\right| \leq \left\| \hat{Y}_i(t) - Y_i \right\| + \left\|\hat{Y}_j(t) - Y_j\right\|$.
    Now, we have $\left|\left\| Y_i - Y_j \right\| - \left\| \hat{Y}_i(t) - \hat{Y}_j(t) \right\|\right|< 2A_1\arcsin\left( \alpha_1 \nu_1 \right)$. Hence we can say $\left\|\hat{Y}_i(t) - \hat{Y}_j(t)\right\|$ converges almost surely to $\left\| Y_i - Y_j \right\|$ for all $i, j \in [M]$. Therefore we have $\Gamma_t$ converges almost surely to $d_H$. Since $\frac{1}{\sin^2(\cdot)}$ is a continuous function, we can say that $\frac{1}{\sin^2\left( \Gamma_t \right)}$ converges almost surely to $\frac{1}{\sin^2\left( d_H \right)}$. Recall that, in the beginning of this proof, we argued that $N\rightarrow \infty$ as $C\rightarrow\infty$. Therefore, we can say that as $C$ tends to $\infty$, we have, $\frac{1}{\sin^2\left( \Gamma_N \right)}$ converges almost surely to $\frac{1}{\sin^2\left( d_H \right)}$. Therefore as $C\rightarrow\infty$, the inequality in equation \eqref{th3: eq: 1} becomes as follows.
    \begin{equation}
        \frac{1}{\sin^2\left( d_H \right)} < \frac{N}{C^2} \leq \frac{1}{\sin^2\left( d_H \right)}.
    \end{equation}
    Hence, as $C\rightarrow\infty$, we have $\frac{N}{C^2}$ converges almost surely to $\frac{1}{\sin^2\left( d_H \right)}$. Now to show $\lim_{C \rightarrow \infty} \mathbb{E}\left[ \left| \frac{N}{C^2} - \frac{1}{\sin^2\left( d_H \right)} \right| \right] = 0$, we have to show $\frac{N}{C^2}$ is uniformly integrable.
    \begin{align}
        \mathbb{E}\left[ \frac{N}{C^2}\mathds{1}\left\{ \frac{N}{C^2}\geq \nu \right\} \right] &\leq \mathbb{E}\left[ \frac{N-\lfloor\nu C^2\rfloor + \nu C^2}{C^2} \mathds{1}\left\{ N \geq \lfloor \nu C^2 \rfloor \right\} \right] \\
        &= \frac{1}{C^2} \mathbb{E}\left[ \left(N-\lfloor\nu C^2\rfloor\right) \mathds{1}\left\{ N \geq \lfloor \nu C^2 \rfloor \right\} \right] +\nu \mathbb{P}\left[ N \geq \lfloor \nu C^2 \rfloor \right] \\
        &= \frac{1}{C^2} \sum_{l=1}^\infty \mathbb{P}\left[ N \geq \lfloor \nu C^2 \rfloor + l \right] + \nu \mathbb{P}\left[ N \geq \lfloor \nu C^2 \rfloor \right] 
    \end{align}
    Assume $\nu$ is large enough to use equation \eqref{th1: eq: 3}. Now by using equation \eqref{th1: eq: 3} on the above equation, we can upper bound as follows. 
    \begin{align}
        &\leq \frac{1}{C^2} \sum_{l=1}^\infty \left\{ M^2\exp{\left( -\alpha_0(\lfloor \nu C^2 \rfloor + l) \right)} + K^2M^4 \exp{\left( \frac{-\eta^2(\lfloor \nu C^2 \rfloor + l)}{16B} \right)} \right\} \\
        & \hspace{2cm} + \nu \left\{ M^2\exp{\left( -\alpha_0(\lfloor \nu C^2 \rfloor) \right)} + K^2M^4 \exp{\left( \frac{-\eta^2(\lfloor \nu C^2 \rfloor)}{16B} \right)} \right\}\\
        &= \frac{M^2}{C^2} \frac{\exp{\left( -\alpha_0(\lfloor \nu C^2 \rfloor + 1) \right)}}{1 - \exp{\left( -\alpha_0 \right)}} + \frac{M^4K^2}{C^2} \frac{\exp{\left( \frac{-\eta^2(\lfloor \nu C^2 \rfloor + l)}{16B} \right)}}{1 - \exp{\left( \frac{-\eta^2}{16B} \right)}} \\
        & \hspace{2cm} + \nu \left\{ M^2\exp{\left( -\alpha_0(\lfloor \nu C^2 \rfloor) \right)} + K^2M^4 \exp{\left( \frac{-\eta^2(\lfloor \nu C^2 \rfloor)}{16B} \right)} \right\}
    \end{align}
    Now, by using the above derived upper bound, it can be verified that $\lim_{\nu\rightarrow\infty}\sup_{C\geq C_0}\mathbb{E}\left[ \frac{N}{C^2}\mathds{1}\left\{ \frac{N}{C^2}\geq \nu \right\} \right]=0$ for some $C_0 \in \mathbb{R}$. This proves the uniform integrability and so we have $\lim_{C \rightarrow \infty} \mathbb{E}\left[ \left| \frac{N}{C^2} - \frac{1}{\sin^2\left( d_H \right)} \right| \right] = 0$. Hence proved.
\end{proof}

Now we prove Theorem \ref{theorem: expconsSEQ}.
\begin{proof}
    From equation \eqref{th2: eq: 4}, for sufficiently large $C$, we have the following upper bound for the error probability.
    \begin{equation}
        \mathbb{P}[E_N] \leq \left( 2Ct_0 + t_0^2 + 1 \right) M^2 \exp{\left( \frac{-(\delta^{''})^2C^2}{64B\alpha_1^2} \right)} + \frac{M^2 \exp{\left( -\alpha_0 C^2 \right)}}{1 - \exp{\left(-\alpha_0\right)}}
    \end{equation}
    It can be verified for sufficiently large $C$, for some $\zeta>0$, probability of error can be bounded as follows.
    \begin{equation}
        \mathbb{P}[E_N] \leq \exp{\left(-(H-\zeta)C^2\right)}, 
    \end{equation}
    where $H = \min\left\{ \frac{(\delta^{''})^2}{64B\alpha_1^2},  \alpha_0 \right\}$. The above equation can be rewritten as $\frac{C^2}{\log\left(\mathbb{P}[E_N]\right)} \geq -\frac{1}{H-\zeta}$. Now on using Theorem \ref{unif}, as $C\rightarrow\infty$, we have the following inequality.
    \begin{equation}
        \frac{N}{C^2} \frac{C^2}{\log\left(\mathbb{P}[E_N]\right)} \geq - \frac{1}{H-\zeta} \frac{1}{\sin^2\left( d_H \right)}.
    \end{equation}
    Now on rearranging and taking expectation, as $C\rightarrow\infty$, we get the following.
    \begin{equation}
        \mathbb{E}[N] \leq -G \log{\mathbb{P}[E_N]}(1+o(1)), \ \ \text{where, } G = \frac{1}{H\sin^2\left( d_H \right)}.
    \end{equation}
    Hence proved.
\end{proof}

\section{Technical results to prove main Theorems}\label{sec: techlemcla}
In this section we drop the index $(t)$ for the estimated quantities to avoid cluttering. It can we verified that the function $K(u) = \exp\left(\frac{-u}{2\sigma^2}\right)$ is a Lipschitz function with the Lipschitz constant $\frac{1}{\sigma\sqrt{e}}$. We use $c$ to denote this constant, that is, $c = \frac{1}{\sigma\sqrt{e}}$. We use $m$ to denote the smallest of the true affinity values between two distributions from the same cluster, and it is more formally defined as $m \coloneqq \min_{k \in [K]} \min_{i, j \in D_k, i \neq j} A_{i, j}$. We define $m_s$ as the minimum norm of the row of the true eigen vector matrix $X$, that is $m_s \coloneqq \min_{i \in [M]} \left\| Z_i \right\|$.

\begin{lemma} \label{lemma: specpt bound}
    Define the event $F \coloneqq \left\{ \forall i, j \in [n], i<j, \left| \hat{d}_{ij} - d_{ij} \right|<\epsilon \right\}$. There exists $\epsilon_3$ such that for all $\epsilon< \epsilon_3$, if the event $F$ holds, then for all $i \in [M]$, we have $\left\| \hat{Y}_i - Y_i \right\| \leq A_1\arcsin\left( \alpha_1 \epsilon \right)$, where $A_1 = \frac{4K}{m_s^2}$ and $\alpha_1 = \frac{12M^4c}{\Delta m^3}$ for some $\Delta>0$.
\end{lemma}
\begin{proof}
    First, we derive bound on the $(i, j)^{th}$ entry of the Lagrange matrix.
    \begin{align}
        \left| \hat{L}_{i, j} - L_{i, j} \right| &= \left| \frac{\overbrace{\hat{A}_{i, j}}^{x_1}}{\underbrace{\left[ \sum_{j_1=1}^M \hat{A}_{i, j_1} \right]^{\frac{1}{2}} \left[ \sum_{i_1=1}^M \hat{A}_{i_1, j} \right]^{\frac{1}{2}}}_{x_2} } - \frac{\overbrace{A_{i, j}}^{y_1}}{\underbrace{\left[ \sum_{j_1=1}^M A_{i, j_1} \right]^{\frac{1}{2}} \left[ \sum_{i_1=1}^M A_{i_1, j} \right]^{\frac{1}{2}}}_{y_2} }  \right| \\
        &\leq  \frac{|y_2||y_1-x_1| + |y_1||y_2-x_2|}{|x_2y_2|} \ \ \text{Claim \ref{claim: x1byx2}} \\
        &\leq \frac{Mc\epsilon + M\frac{\left|x_2^2-y_2^2\right|}{x_2+y_2}}{m(m-c\epsilon)} \\
        &\leq \frac{1}{m(m-c\epsilon)}\left[ Mc\epsilon + \frac{M}{m+m-\epsilon}\left| \underbrace{\left[ \sum_{j_1=1}^M \hat{A}_{i, j_1} \right]}_{x_1} \underbrace{\left[ \sum_{i_1=1}^M \hat{A}_{i_1, j} \right]}_{x_2} - \underbrace{\left[ \sum_{j_1=1}^M A_{i, j_1} \right]}_{y_1} \underbrace{\left[ \sum_{i_1=1}^M A_{i_1, j} \right]}_{y_2} \right| \right] \\
        &\leq \frac{1}{m(m-c\epsilon)}\left[ Mc\epsilon + \frac{M}{m+m-c\epsilon}\left( |x_1 - y_1||x_2| + |x_2 - y_2||y_1| \right) \right] \\
        &\leq \frac{1}{m(m-c\epsilon)}\left[ Mc\epsilon + \frac{M}{m+m-c\epsilon}\left( Mc\epsilon M + Mc\epsilon M \right) \right] \\
        &\leq \frac{1}{m(m-c\epsilon)}\left[ Mc\epsilon + \frac{2M^3c\epsilon}{2m-c\epsilon} \right] \\
        &\leq \frac{3M^3c\epsilon}{m(m-c\epsilon)^2} \\
        &\leq \frac{12M^3c\epsilon}{m^3} \ \ \left(\text{Choose }  \epsilon<\frac{m}{2c}\right).
    \end{align}
    Let $\lambda_j$ be the $j^{th}$ largest eigen value of $L$. Without loss of generality, let say the first $K$ largest eigen values of $L$ has $u$ sets of repeated eigen values. Hence, in the first $K$ largest eigen values, we have only $u$ unique eigen values, where $1\leq u\leq K$. If all the first $K$ eigen values are repeated, then $u=1$ and if all the first $K$ eigen values are distinct $u=K$. Since $L$ is symmetric, it is diagonalizable and hence for all eigen values, algebraic multiplicity is same as geometric multiplicity. Let $m_i$ be the multiplicity (algebraic or geometric) of the $i^{th}$ largest unique eigen value, where $1\leq i\leq u$ and $m_1+\ldots+m_u=K$. Let $\beta_i$ denote the spectral gap of the $i^{th}$ largest unique eigen value and is defined as follows.
\begin{equation}
    \beta_i = \begin{cases}
        \lambda_{m_1} - \lambda_{m_1+1} & \text{if } i = 1 \\
        \min\left\{ \lambda_{m_i} - \lambda_{m_i+1}, \lambda_{m_{i-1}} - \lambda_{m_{i-1}+1} \right\} &\text{for } 2\leq i \leq u.
    \end{cases}
\end{equation}
Let us define the minimum spectral gap $\beta \coloneqq \min_{in \in [u]}\beta_i$. 

First, let us consider the $i^{th}$ eigen space for the sake of discussion for some $i \in [u]$. We write the eigen vector decomposition of $L$ in the following form.
\begin{equation}
    L = \begin{bmatrix}
        U_i & U_{-i}
    \end{bmatrix}
    \begin{bmatrix}
        \Lambda_i & 0 \\
        0 & \Lambda_{-1}
    \end{bmatrix}
    \begin{bmatrix}
        U_i \\
        U_{-i}
    \end{bmatrix},
\end{equation}
where, $\Lambda_i$ is the scalar matrix with the repeated eigen value of the $i^{th}$ eigen space as its scalar and $\Lambda_{-i}$ is the diagonal matrix with the diagonal entries being its remaining eigen values and their corresponding eigen vectors are arranged in the columns of $U_{i}$ and $U_{-i}$ respectively. Similarly, the eigen vector decomposition of $\hat{L}$ can be written as follows.
\begin{equation}
    \hat{L} = \begin{bmatrix}
        \Tilde{U}_i & \Tilde{U}_{-i}
    \end{bmatrix}
    \begin{bmatrix}
        \Tilde{\Lambda}_i & 0 \\
        0 & \Tilde{\Lambda}_{-i}
    \end{bmatrix}
    \begin{bmatrix}
        \Tilde{U}_i \\
        \Tilde{U}_{-i}
    \end{bmatrix},
\end{equation}
where $\Tilde{\Lambda}_i$ consist of the perturbed eigen values corresponding to $\Lambda_i$. 

We derived that each entry of the matrix $\left| \hat{L} - L \right|$ can be upper bounded by $\frac{12M^3c\epsilon}{m^3}$. Hence, its Frobenius norm can be upper bounded by $\left\| \hat{L} - L \right\|_F \leq \frac{12M^4c\epsilon}{m^3}$. Hence from Lemma \ref{lemma: weyl}, we have $\left| \lambda_j\left( \hat{L} \right) - \lambda_j\left( L \right) \right| \leq \frac{12M^4c\epsilon}{m^3}$ for all $j \in [n]$. Let us define $\epsilon_5 = \frac{\beta m^3}{48M^4c}$. Now for all $\epsilon<\epsilon_5$, we have $\left| \lambda_j\left( L \right) - \lambda_j\left( L \right) \right| \leq \frac{\beta}{4}$. Now on applying Lemma \ref{lemma: daviskahan} by choosing $\Delta = \frac{\beta}{2}$, we get 
\begin{equation} \label{eq: davisapplied}
    \sin\left( \Theta(U_i, \Tilde{U}_i) \right) \leq \frac{24M^4c\epsilon}{m^3\beta},
\end{equation} where $\Theta(U_i, \Tilde{U}_i)$ is the diagonal matrix whose entries are the $m_i$ canonical angles between the subspace spanned by the column vectors of $U_i$ and the subspace spanned by the column vector of $\Tilde{U}_i$. 

Let $U_i = [u_1, \ldots, u_{m_i}]$ and $\Tilde{U}_i = \{\Tilde{u}_1, \ldots, \Tilde{u}_{m_i}\}$. Since the subspace space spanned by $U_i$ corresponds to the repeated eigen value, we can chose any collection of orthonormal vectors $u_1, \ldots, u_{m_i}$ in that subspace. Now we make the specific choice as follows.

Let $\{u_1^{'}, \ldots, u_{m_i}^{'}\}$ and $\{\Tilde{u}_1^{'}, \dots, \Tilde{u}_{m_i}^{'}\}$ be the two collections of orthonormal vectors from the subspace corresponding to $U_i$ and $\Tilde{U}_i$ respectively, such that the angle between $u_j^{'}$ and $\Tilde{u}_j^{'}$ forms the $j^{th}$ canonical angle $\theta_j$ for $j \in [m_i]$. Since $\{\Tilde{u}_j, j \in [m_i]\}$ and $\{\Tilde{u}_j^{'}, j \in [m_i]\}$ are the two choice of orthonormal basis for the same subspace, they are related by a orthogonal matrix $Q$ as follows.
\begin{equation}
    \Tilde{u}_j = \sum_{l=1}^{m_i} Q_{lj}\Tilde{u}_l^{'} \text{  for } j = 1, \ldots, m_i.
\end{equation}
Now we choose $u_j, j \in [m_i]$ as follows.
\begin{equation}
    u_j = \sum_{l=1}^{m_i} Q_{lj}u_l^{'} \text{  for } j = 1, \ldots, m_i.
\end{equation}
Now we upper bound the inner product $\langle u_j, \Tilde{u}_j\rangle$ as follows.
\begin{align}
    \langle u_j, \Tilde{u}_j\rangle &= \left\langle \sum_{l=1}^{m_i} Q_{lj}u_l^{'}, \sum_{l=1}^{m_i} Q_{lj}\Tilde{u}_l^{'} \right\rangle \\
    &= \sum_{l=1}^{m_i}\sum_{s=1}^{m_i} Q_{lj}Q_{sj} \langle u_l^{'}, \Tilde{u}_s^{'} \rangle \\
    & = \sum_{l=1}^{m_i} Q_{lj}^2 \langle u_l^{'}, \Tilde{u}_l^{'} \rangle \\
    & = \sum_{l=1}^{m_i} Q_{lj}^2 \cos\theta_l \\
    &\geq \min_{l\in [m_i]} \cos\theta_l
\end{align}

Therefore, the angle between the vectors $u_j$ and $\Tilde{u}_j$, $\angle(u_j, \Tilde{u}_j)$ is upper bounded by $\arccos\left( \min_{l \in [m_i]} \cos\theta_l \right) = \max_{l\in [m_i]} \theta_l$. Since, it holds for all $j \in [m_i]$, we have
\begin{equation}
    \max_{j \in [m_i]} \angle(u_j, \Tilde{u}_j) \leq \max_{l\in [m_i]} \theta_l.
\end{equation}
Let's define $\gamma_i\coloneqq\max_{j \in [m_i]} \angle(u_j, \Tilde{u}_j)$. Therefore, we have $\gamma_i \leq \max_{l\in [m_i]} \theta_l$. Since $\theta_l$ is one of the canonical angles, we get $\sin\theta_l\leq \sin\left( \Theta\left( U_i, \Tilde{U}_i \right) \right)$ and hence we have $\sin\gamma_i\leq \sin\left( \Theta\left( U_i, \Tilde{U}_i \right) \right)$. 
Hence, from equation \eqref{eq: davisapplied}, we get $\sin(\gamma_i) \leq \frac{24M^4c\epsilon}{m^3\beta}$ for all $i \in [u]$. This analysis hold for any $i \in [u]$. Hence we say the maximum of $\sin$ of angles between true and perturbed eigen vector is $\max \sin\gamma_i \leq \sin\left( \Theta\left( U_i, \Tilde{U}_i \right) \right)$.

Since we consider unit norm eigen vectors, it can be verified that the 2-norm of the error vector between the each of the eigen vector of $L$ and its perturbed version $\hat{L}$ is upper bounded by the angle between them. Hence we have $\| \hat{Z}_{:k} - Z_{:k} \| \leq \arcsin{\left( \frac{24M^4c\epsilon}{m^3\beta} \right)}$ for all $k \in [K]$, where the subscript $(:k)$ to indicate that it is the $k^{th}$ column of that matrix. Therefore, the absolute values between the each of the points in the eigen vectors are bounded as $\left| \hat{Z}_{ik} - Z_{ik} \right| \leq \arcsin{\left( \frac{24M^4c\epsilon}{m^3\beta} \right)}$ for $k \in [K]$ and $i \in [M]$. 

First we derive bounds on the one component of the normalized $K$ dimensional spectral point $Y_i = \left[ \frac{Z_{i1}}{\| Z_i \|} \ldots \frac{Z_{iK}}{\| Z_i \|} \right]$.
\begin{align}
    \left| \hat{Y}_{ik} - Y_{ik} \right| &= \left| \frac{\overbrace{\hat{X}_{ik}}^{x_1}}{\underbrace{\| \hat{X}_i \|}_{x_2}} - \frac{\overbrace{X_{ik}}^{y_1}}{\underbrace{\| X_i \|}_{y_2}} \right| \\
    &\leq \frac{|y_2||x_1-y_1| + |y_1||y_2-x_2|}{|x_2||y_2|} \ \ \text{(Claim \ref{claim: x1byx2})}\\
    &\leq \frac{\sqrt{K}\arcsin\left( \left( \frac{24M^4c\epsilon}{m^3\beta} \right) \right) + 1\sqrt{K}\arcsin\left( \left( \frac{24M^4c\epsilon}{m^3\beta} \right) \right)}{m_s\left( m_s -\sqrt{K}\arcsin\left( \left( \frac{24M^4c\epsilon}{m^3\beta} \right) \right) \right)}
\end{align}
Choose $\epsilon_6$ such that $\sqrt{K}\arcsin\left( \left( \frac{12M^4c\epsilon}{m^3\Delta} \right) \right) = \frac{m_s}{2}$. Hence for all $\epsilon < \epsilon_6$, we have, 
\begin{equation}
    \left| \hat{Y}_{ik} - Y_{ik} \right| = \frac{4\sqrt{K}}{m_s^2}\arcsin\left( \left( \frac{24M^4c\epsilon}{m^3\beta} \right) \right)
\end{equation}
Hence we write the bound on the $i^{th}$ spectral point as follows.
\begin{align}
    \left\| \hat{Y}_i - Y_i \right\| &\leq \sqrt{K}\left[ \frac{4\sqrt{K}}{m_s^2}\arcsin\left( \left( \frac{24M^4c\epsilon}{m^3\beta} \right) \right) \right]
\end{align}
 
\begin{align}
    \implies \text{for $\epsilon < \min\left\{\frac{m}{2c}, \epsilon_5, \epsilon_6\right\}$}, \text{we have }\left\| \hat{Y}_i - Y_i \right\| &\leq \frac{4K}{M^2}\arcsin\left( \left( \frac{24M^4c\epsilon}{m^3\beta} \right) \right) \text{ for all $i \in [M]$}.
\end{align}
Hence proved.
\end{proof}

\begin{claim} \label{claim: mmd conc bound}
    Let us define $t_2(\epsilon) \coloneqq \frac{\alpha_2}{\epsilon}$, where $\alpha_2 = 64B$. Then for all $t\geq t_2(\epsilon^2)$, we have
    \begin{equation}
        \mathbb{P}\left[ \exists i, j \in [M], i<j, \left| \hat{d}_{ij}(t) - d_{ij} \right|> \epsilon \right] \leq M^2\exp\left( \frac{-\epsilon^2t}{16B} \right).
    \end{equation}
\end{claim}
\begin{proof}
    Let us define $t_2$ such that $4\sqrt{\frac{B}{t_2}} = \frac{\epsilon}{2}$, that is $t_2 = \frac{64B}{\epsilon^2}$. Hence for all $t\geq t_2$, we have $4\sqrt{\frac{B}{t}} \leq \frac{\epsilon}{2}$. Therefore for all $t\geq t_2$, we can upper bound as follows.
    \begin{align}
        \mathbb{P}\left[ \exists i, j \in [M], i<j, \left| \hat{d}_{ij}(t) - d_{ij} \right|> \epsilon \right] &\leq \mathbb{P}\left[ \exists i, j \in [M], i<j, \left| \hat{d}_{ij}(t) - d_{ij} \right|> 4\sqrt{\frac{B}{t}} + \frac{\epsilon}{2} \right] \\
        &\leq \sum_{i, j \in [M], i<j} \mathbb{P}\left[\left| \hat{d}_{ij}(t) - d_{ij} \right|> 4\sqrt{\frac{B}{t}} + \frac{\epsilon}{2} \right]
    \end{align}
    Now we apply Theorem 7 in \cite{gretton2012kernel}. Hence for all $t\geq t_2$, we can upper bound as follows.
    \begin{align}
        \mathbb{P}\left[ \exists i, j \in [M], i<j, \left| \hat{d}_{ij}(t) - d_{ij} \right|> \epsilon \right] &\leq \sum_{i, j \in [M], i<j} 2 \exp\left( \frac{-\epsilon^2 t}{16B} \right) \\
        &\leq \frac{M^2}{2} 2 \exp\left( \frac{-\epsilon^2 t}{16B} \right) \\
        &= M^2 \exp\left( \frac{-\epsilon^2 t}{16B} \right).
    \end{align}
    Hence proved.
\end{proof}

\begin{claim} \label{claim: delta assump}
    Define $h(D_k)$ for $k \in [K]$ as follows.
    \begin{equation}
        \hat{h}(D_k) \coloneqq \min_{\mathcal{I}\subset D_k} \frac{\sum_{i \in \mathcal{I}, j \in D_k\setminus\mathcal{I}}\hat{A}_{i, j}}{\min\left\{ \sum_{i \in \mathcal{I}}\hat{d}_i^{(k)}, \sum_{j \in D_k\setminus\mathcal{I}}\hat{d}_j^{(k)}  \right\}}.
    \end{equation}
    Let us define $\epsilon_5 = \frac{m}{2c}$. Then for $\epsilon < \epsilon_5$, if the event $F \coloneqq \left\{ \forall i, j \in [M], i<j, \left| \hat{d}_{ij}(t) - d_{ij} \right|<\epsilon \right\}$ holds, then for all $i \in [K]$, we have the following inequality. 
    \begin{equation}
        \frac{\hat{h}(D_k)^2}{2} \geq \delta - A_5\epsilon \ \ \text{for all $i \in [K]$}, \hspace{1cm} \text{where } A_5 = \frac{5M^3c}{8m^3}
    \end{equation}
\end{claim}
\begin{proof}
    \begin{align}
        \left| \frac{\hat{h}(D_k)^2}{2} - \frac{h(D_k)^2}{2} \right| &= \frac{1}{2} \left| \left[ \min_{\mathcal{I}\subset D_k} \frac{\sum_{i \in \mathcal{I}, j \in D_k\setminus\mathcal{I}}\hat{A}_{i, j}}{\min\left\{ \sum_{i \in \mathcal{I}}\hat{d}_i^{(k)}, \sum_{j \in D_k\setminus\mathcal{I}}\hat{d}_j^{(k)}  \right\}} \right]^2 - \left[ \min_{\mathcal{I}\subset D_k} \frac{\sum_{i \in \mathcal{I}, j \in D_k\setminus\mathcal{I}}A_{i, j}}{\min\left\{ \sum_{i \in \mathcal{I}}d_i^{(k)}, \sum_{j \in D_k\setminus\mathcal{I}}d_{j}^{(k)}  \right\}} \right]^2 \right|
    \end{align}
    It can be verified that the above expression can be upper bounded for some $\mathcal{I}\subset D_k$. 
    \begin{align}
        \left| \frac{\hat{h}(D_k)^2}{2} - \frac{h(D_k)^2}{2} \right| &\leq \frac{1}{2} \left| \left[ \frac{\sum_{i \in \mathcal{I}, j \in D_k\setminus\mathcal{I}}\hat{A}_{i, j}}{\min\left\{ \sum_{i \in \mathcal{I}}\hat{d}_i^{(k)}, \sum_{j \in D_k\setminus\mathcal{I}}\hat{d}_j^{(k)}  \right\}} \right]^2 - \left[ \frac{\sum_{i \in \mathcal{I}, j \in D_k\setminus\mathcal{I}}A_{i, j}}{\min\left\{ \sum_{i \in \mathcal{I}}d_i^{(k)}, \sum_{j \in D_k\setminus\mathcal{I}}d_{j}^{(k)}  \right\}} \right]^2 \right| \\
        &= \frac{1}{2} \left| \left[ \frac{\sum_{i \in \mathcal{I}, j \in D_k\setminus\mathcal{I}}\hat{A}_{i, j}}{\min\left\{ \sum_{i \in \mathcal{I}}\hat{d}_i^{(k)}, \sum_{j \in D_k\setminus\mathcal{I}}\hat{d}_j^{(k)}  \right\}} \right] + \left[ \frac{\sum_{i \in \mathcal{I}, j \in D_k\setminus\mathcal{I}}A_{i, j}}{\min\left\{ \sum_{i \in \mathcal{I}}d_i^{(k)}, \sum_{j \in D_k\setminus\mathcal{I}}d_{j}^{(k)}  \right\}} \right] \right| \\
        &\hspace{1cm} \left| \left[ \frac{\sum_{i \in \mathcal{I}, j \in D_k\setminus\mathcal{I}}\hat{A}_{i, j}}{\min\left\{ \sum_{i \in \mathcal{I}}\hat{d}_i^{(k)}, \sum_{j \in D_k\setminus\mathcal{I}}\hat{d}_j^{(k)}  \right\}} \right] - \left[ \frac{\sum_{i \in \mathcal{I}, j \in D_k\setminus\mathcal{I}}A_{i, j}}{\min\left\{ \sum_{i \in \mathcal{I}}d_i^{(k)}, \sum_{j \in D_k\setminus\mathcal{I}}d_{j}^{(k)}  \right\}} \right] \right|.
    \end{align}
    Now we upper and lower bound the numerator and denominator respectively of the first term of the product in the above expression.
    \begin{align}
        &\left| \frac{\hat{h}(D_k)^2}{2} - \frac{h(D_k)^2}{2} \right| \\
        &\leq \frac{1}{2}\left| \frac{\frac{M^2}{4}}{M_km} + \frac{\frac{M^2}{4}}{M_k(m-c\epsilon)} \right| \left| \left[ \frac{\sum_{i \in \mathcal{I}, j \in D_k\setminus\mathcal{I}}\hat{A}_{i, j}}{\min\left\{ \sum_{i \in \mathcal{I}}\hat{d}_i^{(k)}, \sum_{j \in D_k\setminus\mathcal{I}}\hat{d}_j^{(k)}  \right\}} \right] - \left[ \frac{\sum_{i \in \mathcal{I}, j \in D_k\setminus\mathcal{I}}A_{i, j}}{\min\left\{ \sum_{i \in \mathcal{I}}d_i^{(k)}, \sum_{j \in D_k\setminus\mathcal{I}}d_{j}^{(k)}  \right\}} \right] \right| \\
        &\leq \frac{M^2}{4M_k(m-c\epsilon)} \left| \left[ \frac{\overbrace{\sum_{i \in \mathcal{I}, j \in D_k\setminus\mathcal{I}}\hat{A}_{i, j}}^{x_1}}{\underbrace{\min\left\{ \sum_{i \in \mathcal{I}}\hat{d}_i^{(k)}, \sum_{j \in D_k\setminus\mathcal{I}}\hat{d}_j^{(k)}  \right\}}_{x_2}} \right] - \left[ \frac{\overbrace{\sum_{i \in \mathcal{I}, j \in D_k\setminus\mathcal{I}}A_{i, j}}^{y_1}}{\underbrace{\min\left\{ \sum_{i \in \mathcal{I}}d_i^{(k)}, \sum_{j \in D_k\setminus\mathcal{I}}d_{j}^{(k)}  \right\}}_{y_2}} \right] \right| \\
        &\leq \frac{M^2}{4M_k(m-c\epsilon)} \left[ \frac{|y_2||x_1 - y_1| + |y_1||y_2 - x_2|}{|x_2y_2|} \right] \text{  (Claim \ref{claim: x1byx2})}\\
        &\leq \frac{M^2}{4M_k(m-c\epsilon)} \left[ \frac{\left(\frac{M_k}{2}M_k\right)\frac{M_k^2}{4}c\epsilon + \left(\frac{M_k}{2}M_k\right)|y_2 - x_2|}{M_kmM_k(m-c\epsilon)} \right] \\
        &\leq \frac{M^2}{4M_km(m-c\epsilon)^2} \left[ \frac{M_k^2}{8}c\epsilon + \frac{1}{2}\left| \min\left\{ \underbrace{\sum_{i \in \mathcal{I}}\hat{d}_i^{(k)}}_{x_1}, \underbrace{\sum_{j \in D_k\setminus\mathcal{I}}\hat{d}_j^{(k)}}_{x_2}  \right\} - \min\left\{ \underbrace{\sum_{i \in \mathcal{I}}d_i^{(k)}}_{y_1}, \underbrace{\sum_{j \in D_k\setminus\mathcal{I}}d_{j}^{(k)}}_{y_2}  \right\} \right| \right] \\
        &\leq \frac{M^2}{4M_km(m-c\epsilon)^2} \left[ \frac{M_k^2}{8}c\epsilon + \frac{1}{2}\left( |y_1-x_1| + |y_2-x_2| \right) \right] \text{  (Claim \ref{claim: minx1x2})}\\
        &\leq \frac{M^2}{4M_km(m-c\epsilon)^2} \left[ \frac{M_k^2}{8}c\epsilon + \frac{1}{2}\left( \frac{M_k}{2}M_kc\epsilon + \frac{M_k}{2}M_kc\epsilon \right) \right] \\
        &\leq \frac{M^2}{4M_km(m-c\epsilon)^2} \left[ \frac{5}{8}M_k^2c\epsilon \right] \\
        &= \frac{5M^2M_kc\epsilon}{32m(m-c\epsilon)^2} \\
        &\leq \frac{5M^3c\epsilon}{32m(m-c\epsilon)^2} \\
        &\leq \frac{5M^3c\epsilon}{8m^3} \hspace{1cm} \left( \because \epsilon < \frac{m}{2c} \right).
    \end{align}
    Therefore we have the following inequality.
    \begin{align}
        \frac{\hat{h}(D_k)^2}{2} &\geq \frac{h(D_k)^2}{2} - \frac{5M^3c\epsilon}{8m^3} \\
        &\geq \delta - \frac{5M^3c\epsilon}{8m^3}. \text{   (Assumption \ref{ass: 1})}
    \end{align}
    Hence proved.
\end{proof}

\begin{claim} \label{claim: epsilon1 assump}
    Let us define $\epsilon_5 = \frac{m}{2c}$. Then for $\epsilon<\epsilon_5$, if the event $F \coloneqq \left\{ \forall i, j \in [n], i<j, \left| d_{ij}(t) - d_{ij} \right|<\epsilon \right\}$ holds, then for all $k_i, k_2 \in [K], k_1\neq k_2$, we have the following inequality.
    \begin{equation}
        \sum_{i \in D_{k_1}} \sum_{j \in D_{k_2}} \frac{\hat{A}_{i, j}^2}{\hat{d}_i^{(k_1)}\hat{d}_j^{(k_2)}} \leq \epsilon_1 + A_6\epsilon, \hspace{1cm} \text{where, } A_6 = \frac{24c}{m^4}.
    \end{equation}
\end{claim}
\begin{proof}
    \begin{align}
        &\left| \sum_{i \in D_{k_1}} \sum_{j \in D_{k_2}} \frac{\hat{A}_{i, j}^2}{\hat{d}_i^{(k_1)}\hat{d}_j^{(k_2)}} - \sum_{i \in D_{k_1}} \sum_{j \in D_{k_2}} \frac{A_{i, j}^{2}}{d_i^{(k_1)}d_j^{(k_2)}} \right| \\
        &\leq \sum_{i \in D_{k_1}} \sum_{j \in D_{k_2}} \left| \frac{\overbrace{\hat{A}_{i, j}^2}^{x_1}}{\underbrace{\hat{d}_i^{(k_1)}}_{x_2}\underbrace{\hat{d}_j^{(k_2)}}_{x_3}} - \frac{\overbrace{A_{i, j}^{2}}^{y_1}}{\underbrace{d_i^{(k_1)}}_{y_2}\underbrace{d_j^{(k_2)}}_{y_3}} \right| \\
        &\leq M_{k_1}M_{k_2} \left[ \frac{|x_1y_3||y_2-x_2| + |x_1x_2||y_3-x_3| + 2|x_1x_3||y_2-x_2| + |x_2x_3||y_1-x_1|}{|x_2x_3y_2y_3|} \right] \text{  (Claim \ref{claim: x1byx2x3})}\\
        &\leq M_{k_1}M_{k_2} \left[ \frac{M_{k_1}M_{k_2}c\epsilon + M_{k_1}M_{k_2}c\epsilon + 2M_{k_2}M_{k_1}c\epsilon + 2c\epsilon M_{k_1}M_{k_2}}{M_{k_1}mM_{k_2}mM_{k_1}(m-\epsilon)M_{k_2}(m-c\epsilon)} \right] \\
        &= \frac{6c\epsilon}{m^2(m-c\epsilon)^2} \\
        &\leq \frac{24c\epsilon}{m^4} \hspace{1cm} \left(\because \epsilon <\frac{m}{2c}\right).
    \end{align}
    Therefore we have the following inequality.
    \begin{align}
        \sum_{i \in D_{k_1}} \sum_{j \in D_{k_2}} \frac{\hat{A}_{i, j}^2}{\hat{d}_i^{(k_1)}\hat{d}_j^{(k_2)}} &\leq \sum_{i \in D_{k_1}} \sum_{j \in D_{k_2}} \frac{A_{i, j}^{2}}{d_i^{(k_1)}d_j^{(k_2)}} + \frac{24c\epsilon}{m^4} \\
        &\leq \epsilon_1 + \frac{24c\epsilon}{m^4}. \text{    (Assumption \ref{ass: 2})}
    \end{align}
    Hence proved.
\end{proof}

\begin{claim} \label{claim: epsilon2 assump}
    Let us define $\epsilon_5 = \frac{m}{2c}$. Then for $\epsilon < \epsilon_5$, if the event $F \coloneqq \left\{ \forall i, j \in [n], i<j, \left| d_{ij} - d_{ij} \right|<\epsilon \right\}$ holds, then for all $k \in [K]$ and $i \in D_k$, we have the following inequality.
    \begin{equation}
        \left[ \sum_{j \notin D_k} \frac{\hat{A}_{i, j}}{\hat{d}_i^{(k)}} \right] \left[ \sum_{i_1, j_1 \in D_k}\frac{\hat{A}_{i_1, j_1}^2}{\hat{d}_{i_1}^{(k)}, \hat{d}_{j_1}^{(k)}} \right]^{\frac{1}{2}} \leq \epsilon_2 + A_7\epsilon, \ \ \text{where } A_7=\frac{80Mc}{m^7}.
    \end{equation}
\end{claim}
\begin{proof}
    \begin{align}
        &\left| \left[ \sum_{j \notin D_k} \frac{\hat{A}_{i, j}}{\hat{d}_i^{(k)}} \right] \left[ \sum_{i_1, j_1 \in D_k}\frac{\hat{A}_{i_1, i_2}^2}{\hat{d}_{i_1}^{(k)}, \hat{d}_{j_1}^{(k)}} \right]^{\frac{1}{2}} - \left[ \sum_{j \notin D_k} \frac{A_{i, j}}{d_i^{(k)}} \right] \left[ \sum_{i_1, j_1 \in D_k}\frac{A_{i_1, i_2}^{2}}{d_{i_1}^{(k)} d_{j_1}^{(k)}} \right]^{\frac{1}{2}} \right| \\
        &= \left| \frac{\overbrace{\left[ \sum_{j \notin D_k} \hat{A}_{i, j} \right] \left[ \sum_{i_1, j_1 \in D_k}\frac{\hat{A}_{i_1, i_2}^2}{\hat{d}_{i_1}^{(k)}, \hat{d}_{j_1}^{(k)}} \right]^{\frac{1}{2}}}^{x_1}}{\underbrace{\hat{d}_i^{(k)}}_{x_2}} - \frac{\overbrace{\left[ \sum_{j \notin D_k} A_{i, j} \right] \left[ \sum_{i_1, j_1 \in D_k}\frac{A_{i_1, i_2}^{2}}{d_{i_1}^{(k)} d_{j_1}^{(k)}} \right]^{\frac{1}{2}}}^{y_1}}{\underbrace{d_i^{(k)}}_{y_2}} \right| \\
        &\leq \frac{|y_2||y_1-x_1| + |y_1||y_2-x_2|}{|x_2y_2|} \\
        &\leq \frac{M_k|y_1-x_1| + (M-M_k)\left( M_k^2\frac{1}{mM_k mM_k} \right)^{\frac{1}{2}}M_kc\epsilon}{M_k(m-c\epsilon)M_km} \\
        &= \frac{1}{M_km(m-c\epsilon)}\left| \underbrace{\left[ \sum_{j \notin D_k} \hat{A}_{i, j} \right]}_{x_1} \underbrace{\left[ \sum_{i_1, j_1 \in D_k}\frac{\hat{A}_{i_1, i_2}^2}{\hat{d}_{i_1}^{(k)}, \hat{d}_{j_1}^{(k)}} \right]^{\frac{1}{2}}}_{x_2} - \underbrace{\left[ \sum_{j \notin D_k} A_{i, j} \right]}_{y_1} \underbrace{\left[ \sum_{i_1, j_1 \in D_k}\frac{A_{i_1, i_2}^{2}}{d_{i_1}^{(k)} d_{j_1}^{(k)}} \right]^{\frac{1}{2}}}_{y_2} \right| + \frac{(M-M_k)c\epsilon}{M_km^2(m-c\epsilon)} \\
        &\leq \frac{1}{M_km(m-c\epsilon)}\left[ |x_1-y_1||x_2| + |x_2-y_2||y_1| \right] + \frac{(M-M_k)c\epsilon}{M_km^2(m-c\epsilon)} \\
        &\leq \frac{1}{M_km(m-c\epsilon)}\left[ (M-M_k)c\epsilon\left( M_k^2\frac{1}{(m-c\epsilon)M_k(m-c\epsilon)M_k} \right)^{\frac{1}{2}} + |x_2-y_2|(M-M_k) \right] + \frac{(M-M_k)c\epsilon}{M_km^2(m-c\epsilon)} \\
        &= \frac{(M-M_k)c\epsilon}{M_km(m-c\epsilon)^2} + \frac{(M-M_k)}{M_km(m-c\epsilon)}\left|x_2-y_2\right| + \frac{(M-M_k)c\epsilon}{M_km^2(m-c\epsilon)} \\
        &\leq 2\frac{(M-M_k)c\epsilon}{M_km(m-c\epsilon)^2} + \frac{(M-M_k)}{M_km(m-c\epsilon)} \left| \underbrace{\left[ \sum_{i_1, j_1 \in D_k}\frac{\hat{A}_{i_1, i_2}^2}{\hat{d}_{i_1}^{(k)}, \hat{d}_{j_1}^{(k)}} \right]^{\frac{1}{2}}}_{\sqrt{x_1}}-\underbrace{\left[ \sum_{i_1, j_1 \in D_k}\frac{A_{i_1, i_2}^{2}}{d_{i_1}^{(k)} d_{j_1}^{(k)}} \right]^{\frac{1}{2}}}_{\sqrt{y_1}}\right| \\
        &= 2\frac{(M-M_k)c\epsilon}{M_km(m-c\epsilon)^2} + \frac{(M-M_k)}{M_km(m-c\epsilon)}\frac{|x_1-y_1|}{\sqrt{x_1}+\sqrt{y_1}}
    \end{align}
    \begin{align}
        &\leq 2\frac{(M-M_k)c\epsilon}{M_km(m-c\epsilon)^2} + \frac{(M-M_k)}{M_km(m-c\epsilon)}\frac{|x_1-y_1|}{\sqrt{M_k^2\frac{m^2}{M_kM_k}}+\sqrt{M_k^2\frac{(m-c\epsilon)^2}{M_kM_k}}} \\
        &\leq 2\frac{(M-M_k)c\epsilon}{M_km(m-c\epsilon)^2} + \frac{(M-M_k)}{2M_km(m-c\epsilon)^2}  \left|  \sum_{i_1, j_1 \in D_k}\frac{\hat{A}_{i_1, i_2}^2}{\hat{d}_{i_1}^{(k)} \hat{d}_{j_1}^{(k)}} - \sum_{i_1, j_1 \in D_k}\frac{A_{i_1, i_2}^{2}}{d_{i_1}^{(k)} d_{j_1}^{(k)}} \right| \\
        &\leq 2\frac{(M-M_k)c\epsilon}{M_km(m-c\epsilon)^2} + \frac{(M-M_k)}{2M_km(m-c\epsilon)^2}  \sum_{i_1, j_1 \in D_k} \left|  \frac{\overbrace{\hat{A}_{i_1, i_2}^2}^{x_1}}{\underbrace{\hat{d}_{i_1}^{(k)}}_{x_2} \underbrace{\hat{d}_{j_1}^{(k)}}_{x_3}} - \frac{\overbrace{A_{i_1, i_2}^{2}}^{y_1}}{\underbrace{d_{i_1}^{(k)}}_{y_2} \underbrace{d_{j_1}^{(k)}}_{y_3}} \right| \\
        \intertext{Now by using Claim \ref{claim: x1byx2x3}, we can upper bound as follows.}
        &\leq 2\frac{(M-M_k)c\epsilon}{M_km(m-c\epsilon)^2} + \frac{(M-M_k)}{2M_km(m-c\epsilon)^2}  \sum_{i_1, j_1 \in D_k} \left[ \frac{|x_1y_3||y_2-x_2| + |x_1x_2||y_3-x_3| + 2|x_1x_3||y_2-x_2| + |x_2x_3||y_1-x_1|}{|x_2x_3y_2y_3|} \right] \\
        &\leq 2\frac{(M-M_k)c\epsilon}{M_km(m-c\epsilon)^2} + \frac{(M-M_k)}{2M_km(m-c\epsilon)^2} M_k^2 \frac{M_kM_kc\epsilon + M_kM_kc\epsilon + 2M_kM_kc\epsilon + 2c\epsilon M_kM_k}{M_k(m-c\epsilon)M_k(m-c\epsilon)M_kmM_km} \\
        &\leq 2\frac{(M-M_k)c\epsilon}{M_km(m-c\epsilon)^2} + \frac{(M-M_k)}{2M_km^3(m-c\epsilon)^4}6c\epsilon \\
        &\leq 2\frac{(M-M_k)c\epsilon}{M_km^3(m-c\epsilon)^4} + \frac{(M-M_k)}{2M_km^3(m-c\epsilon)^4}6c\epsilon \\
        &= 5\frac{(M-M_k)c\epsilon}{M_km^3(m-c\epsilon)^4} \\
        &\leq \frac{80Mc\epsilon}{m^7} \ \ \left( \because M_k\geq 1, \epsilon \leq \frac{m}{2c} \right)\\
    \end{align}
    Therefore we have the following inequality.
    \begin{align}
        \left[ \sum_{j \notin D_k} \frac{\hat{A}_{i, j}}{\hat{d}_i^{(k)}} \right] \left[ \sum_{i_1, j_1 \in D_k}\frac{\hat{A}_{i_1, i_2}^2}{\hat{d}_{i_1}^{(k)}, \hat{d}_{j_1}^{(k)}} \right]^{\frac{1}{2}} &\leq \left[ \sum_{j \notin D_k} \frac{A_{i, j}}{d_i^{(k)}} \right] \left[ \sum_{i_1, j_1 \in D_k}\frac{A_{i_1, i_2}^{2}}{d_{i_1}^{(k)} d_{j_1}^{(k)}} \right]^{\frac{1}{2}} + \frac{80Mc\epsilon}{m^7} \\
        &\leq \epsilon_2 + \frac{80Mc\epsilon}{m^7}. \text{    (Assumption \ref{ass: 3})}
    \end{align}
    Hence proved.
\end{proof}

\begin{claim} \label{claim: C assump}
    Let us define $\epsilon_5=\frac{m}{2c}$. Then for $\epsilon<\epsilon_5$, if the event $F \coloneqq \left\{ \forall i, j \in [n], i<j, \left| d_{ij} - d_{ij} \right|<\epsilon \right\}$ holds, then for all $k \in [K]$ and $i \in D_k$, we have the following inequality.
    \begin{equation}
        \frac{1}{M_k}\frac{\sum_{i_1 \in D_k}\hat{d}_{i_1}^{(k)}}{\hat{d}_i^{(k)}} \leq C + A_8\epsilon, \ \ \text{where } A_8 = \frac{4c}{m^2}.
    \end{equation}
\end{claim}
\begin{proof}
    \begin{align}
        \left| \frac{1}{M_k}\frac{\sum_{i_1 \in D_k}\hat{d}_{i_1}^{(k)}}{\hat{d}_i^{(k)}} - \frac{1}{M_k}\frac{\sum_{i_1 \in D_k}d_{i_1}^{(k)}}{d_i^{(k)}} \right| &\leq \frac{1}{M_k} \sum_{i_1 \in D_k} \left| \frac{\overbrace{\hat{d}_{i_1}^{(k)}}^{x_1}}{\underbrace{\hat{d}_i^{(k)}}_{x_2}} - \frac{\overbrace{d_{i_1}^{(k)}}^{y_1}}{\underbrace{d_i^{(k)}}_{y_2}} \right| \\
        &\leq \frac{1}{M_k} \sum_{i_1 \in D_k}\left[\frac{|y_2||y_1-x_1| + |y_1||y_2-x_2|}{|x_2y_2|} \right] \ \ \text{(Claim \ref{claim: x1byx2})}\\
        &\leq \frac{1}{M_k} \sum_{i_1 \in D_k}\left[\frac{M_kM_kc\epsilon + M_kM_kc\epsilon}{M_k(m-c\epsilon)M_km} \right] \\
        &\leq \frac{2c\epsilon}{m(m-\epsilon)} \ \ \left(\because |D_k|=M_k\right) \\
        &\leq \frac{4c\epsilon}{m^2} \ \ \left(\because \epsilon < \frac{m}{2c}\right).
    \end{align}
    Therefore, we have the following inequality.
    \begin{align}
        \frac{1}{M_k}\frac{\sum_{i_1 \in D_k}\hat{d}_{i_1}^{(k)}}{\hat{d}_i^{(k)}} &\leq \frac{1}{M_k}\frac{\sum_{i_1 \in D_k}d_{i_1}^{(k)}}{d_i^{(k)}} + \frac{4c\epsilon}{m^2} \\
        &\leq C + \frac{4c\epsilon}{m^2}. \text{    (Assumption \ref{ass: 4})}
    \end{align}
    Hence proved.
\end{proof}

\section{Supplementary Mathematical Results}
\begin{claim} \label{claim: x1byx2x3}
    Let us define the function $f(x_1, x_2, x_3) = \frac{x_1}{x_2x_3}$. Then we can upper bound $\left| f(x_1, x_2, x_3) - f(y_1, y_2, y_3) \right|$ as follows.
    \begin{equation}
        \left| f(x_1, x_2, x_3) - f(y_1, y_2, y_3) \right| \leq \frac{|x_1y_3||y_2-x_2| + |x_1x_2||y_3-x_3| + 2|x_1x_3||y_2-x_2| + |x_2x_3||y_1-x_1|}{|x_2x_3y_2y_3|}.
    \end{equation}
\end{claim}
\begin{proof}
    \begin{align}
        \left| f(x_1, x_2, x_3) - f(y_1, y_2, y_3) \right| &= \left| \frac{x_1}{x_2x_3} - \frac{y_1}{y_2y_3} \right| \\
        &= \left| \frac{x_1y_2y_3 - y_1x_2x_3}{x_2x_3y_2y_3} \right| \\
        &= \left| \frac{x_1y_2y_3 - x_1x_2y_3 + x_1x_2y_3 - x_1y_2x_3 + x_1y_2x_3 - y_1x_2x_3}{x_2x_3y_2y_3} \right| \\
        &= \left| \frac{x_1y_3(y_2 - x_2) + x_1(x_2y_3 - y_2x_3) + x_3(x_1y_2 - y_1x_2)}{x_2x_3y_2y_3} \right| \\
        &= \left| \frac{x_1y_3(y_2 - x_2) + x_1(x_2y_3 - x_2x_3 + x_2x_3 - y_2x_3) + x_3(x_1y_2 - x_1x_2 + x_1x_2 - y_1x_2)}{x_2x_3y_2y_3} \right| \\
        &= \left| \frac{x_1y_3(y_2 - x_2) + x_1x_2(y_3 - x_3) + x_1x_3(x_2 - y_2) + x_3x_1(y_2 - x_2) + x_3x_2(x_1 - y_1)}{x_2x_3y_2y_3} \right| \\
        &\leq \frac{|x_1y_3||y_2 - x_2| + |x_1x_2||y_3 - x_3| + 2|x_1x_3||x_2 - y_2| + |x_3x_2||x_1 - y_1|}{|x_2x_3y_2y_3|}.
    \end{align}
\end{proof}

\begin{claim} \label{claim: x1byx2}
    Let us define the function $f(x_1, x_2) = \frac{x_1}{x_2}$. Then we can upper bound $\left| f(x_1, x_2) - f(y_1, y_2) \right|$ as follows.
    \begin{equation}
        \left| f(x_1, x_2) - f(y_1, y_2) \right| \leq \frac{|y_2||y_1-x_1| + |y_1||y_2-x_2|}{|x_2y_2|}.
    \end{equation}
\end{claim}
\begin{proof}
    \begin{align}
        \left| f(x_1, x_2) - f(y_1, y_2) \right| &= \left| \frac{x_1}{x_2} - \frac{y_1}{y_2} \right| \\
        &= \left| \frac{x_1y_2 - x_2y_1}{x_2y_2} \right| \\
        &= \left| \frac{x_1y_2 - y_1y_2 + y_1y_2 - x_2y_1}{x_2y_2} \right| \\
        &= \left| \frac{y_2(x_1 - y_1) + y_1(y_2 - x_2)}{x_2y_2} \right| \\
        &\leq \frac{|y_2||x_1 - y_1| + |y_1||y_2 - x_2|}{|x_2y_2|}.
    \end{align}
\end{proof}

\begin{claim} \label{claim: minx1x2}
    Let us define the function $f(x_1, x_2) = \min\{ x_1, x_2 \}$. Then we can upper bound $\left| f(x_1, x_2) - f(y_1, y_2) \right|$ as follows.
    \begin{equation}
        \left| f(x_1, x_2) - f(y_1, y_2) \right| \leq |x_1 - y_1| + |x_2 - y_2|.
    \end{equation}
\end{claim}
\begin{proof}
    It can be verified that $\min\{ x_1, x_2 \} = \frac{x_1+x_2}{2} - \frac{|x_1-x_2|}{2}$.
    \begin{align}
        \left| f(x_1, x_2) - f(y_1, y_2) \right| &= \left| \min\{ x_1, x_2 \} - \min\{ y_1, y_2 \} \right| \\
        &= \left| \frac{x_1+x_2}{2} - \frac{|x_1-x_2|}{2} - \frac{y_1+y_2}{2} + \frac{|y_1-y_2|}{2} \right| \\
        &\leq \frac{\left| (x_1+x_2) - (y_1+y_2) \right|}{2} + \frac{\left| |x_1-x_2| - |y_1-y_2| \right|}{2} \\
        &\leq \frac{|x_1-y_1| + |x_2-y_2|}{2} + \frac{\left| (x_1-x_2) - (y_1-y_2) \right|}{2} \\
        &\leq \frac{|x_1-y_1| + |x_2-y_2|}{2} + \frac{|x_1-y_1| + |x_2-y_2|}{2} \\
        &= |x_1-y_1| + |x_2-y_2|.
    \end{align}
\end{proof}

\begin{claim} \label{claim: x1x2}
    Let us define the function $f(x_1, x_2) = x_1x_2$. Then we can upper bound $\left| f(x_1, x_2) - f(y_1, y_2) \right|$ as follows.
    \begin{equation}
        \left| f(x_1, x_2) - f(y_1, y_2) \right| \leq |x_1 - y_1||x_2| + |x_2 - y_2||y_1|.
    \end{equation}
\end{claim}
\begin{proof}
    \begin{align}
        \left| f(x_1, x_2) - f(y_1, y_2) \right| &= \left| x_1x_2 - y_1y_2 \right| \\
        &= \left| x_1x_2 - y_1x_2 + y_1x_2 - y_1y_2 \right| \\
        &= \left| (x_1 - y_1)x_2 + y_1(x_2 - y_2) \right| \\
        &\leq |x_1 - y_1||x_2| + |y_1||x_2 - y_2|.
    \end{align}
\end{proof}

\begin{claim} \label{claim: boldempsilon bound}
    Define $\boldsymbol{\epsilon}\coloneqq \sqrt{K(K-1)\epsilon_1 + K\epsilon_2^2}$ and $\boldsymbol{\epsilon}^{'} \coloneqq \sqrt{K(K-1)\left[ \epsilon_1 + u(\epsilon) \right] + K\left[ \epsilon_2 + v(\epsilon) \right]^2}$, where $u(\epsilon)=A_6 \epsilon$ and $v(\epsilon) = A_7\epsilon$ for some $A_6, A_7>0$ with $\epsilon<\epsilon_T$ for some $\epsilon_T > 0$. 
    It can be shown that $\boldsymbol{\epsilon}^{'} \leq \boldsymbol{\epsilon} + \Tilde{C}(\epsilon_T)\epsilon$, for some $\Tilde{C}(\epsilon_T)>0$.
\end{claim}
\begin{proof}
    \begin{align}
        \left| \boldsymbol{\epsilon}^{'} - \boldsymbol{\epsilon} \right| &= \sqrt{K}\left| \sqrt{(K-1)(\epsilon_1+u(\epsilon))+(\epsilon_2+v(\epsilon))^2} - \sqrt{(K-1)\epsilon_1+\epsilon_2^2}\right|.
    \end{align}
    Let us define the function $f(x_1, x_2) = \sqrt{ax_1 + x_2^2}$, where $a=K-1$. The partial derivative of $f$ with respect to $x_1$ is given by $\frac{\partial f(x_1, x_2)}{\partial x_1} = \frac{a}{2\sqrt{ax_1+x_2^2}}$. Since $\frac{\partial f(x_1, x_2)}{\partial x_1}$ is decreasing both in $x_1$ and $x_2$, we can say $\frac{\partial f(x_1, x_2)}{\partial x_1} \leq \frac{\partial f(\epsilon_1, \epsilon_2)}{\partial x_1}$ in the interval $\epsilon_1\leq x_1 \leq \epsilon_1+u(\epsilon), \epsilon_2\leq x_2 \leq \epsilon_2+v(\epsilon)$. The partial derivative of $f$ with respect to $x_2$ is given by $\frac{\partial f(x_1, x_2)}{\partial x_2} = \frac{2x_2}{2\sqrt{ax_1+x_2^2}}$. Since $\frac{\partial f(x_1, x_2)}{\partial x_2}$ is decreasing in $x_1$ and increasing in $x_2$ and using the fact that $\epsilon<\epsilon_T$, we can say $\frac{\partial f(x_1, x_2)}{\partial x_1} \leq \frac{\partial f(\epsilon_1, \epsilon_2+v(\epsilon_T))}{\partial x_2}$ in the interval $\epsilon_1\leq x_1 \leq \epsilon_1+u(\epsilon), \epsilon_2\leq x_2 \leq \epsilon_2+v(\epsilon)$. Hence, we can say that $f(x_1, x_2)$ is a Lipschitz function in the interval $\epsilon_1\leq x_1 \leq \epsilon_1+u(\epsilon), \epsilon_2\leq x_2 \leq \epsilon_2+v(\epsilon)$ with Lipschitz constant $L(\epsilon_T) = \sqrt{\left( \frac{\partial f(\epsilon_1, \epsilon_2)}{\partial x_1} \right)^2 + \left( \frac{\partial f(\epsilon_1, \epsilon_2+v(\epsilon_T))}{\partial x_2} \right)^2}$. 
    \begin{align}
        \left| \boldsymbol{\epsilon}^{'} - \boldsymbol{\epsilon} \right| &\leq \sqrt{K}L(\epsilon_T) \sqrt{u(\epsilon)^2 + v(\epsilon)^2} \\ 
        &\leq \sqrt{K}L(\epsilon_T) \sqrt{\left( A_6\epsilon \right)^2 + \left( A_7\epsilon \right)^2} .
    \end{align}
    Therefore we have $\boldsymbol{\epsilon}^{'} \leq \boldsymbol{\epsilon} + \Tilde{C}(\epsilon_T)\epsilon$, where $\Tilde{C}(\epsilon_T) = \sqrt{K}L(\epsilon_T) \sqrt{\left( A_6 \right)^2 + \left( A_7 \right)^2}$. Hence proved.
\end{proof}

\section{Auxiliary results from literature}

\begin{lemma} \label{lemma: ngTheorem2}
    Let assumptions \ref{ass: 1}, \ref{ass: 2}, \ref{ass: 3} and \ref{ass: 4} hold. Set $\boldsymbol{\epsilon} = \sqrt{K(K-1)\epsilon_1+K\epsilon_2^2}$. If $\delta>(2+\sqrt{2})\boldsymbol{\epsilon}$, then there exist $K$ orthonormal vectors $r_1, \ldots, r_K$ so that rows of the matrix $Y$ satisfy
    \begin{equation}
        \frac{1}{M}\sum_{k=1}^K \sum_{i \in D_k}\left\| Y_i - r_k \right\|_2^2 \leq 4C\left( 4+2\sqrt{K} \right)^2 \frac{\boldsymbol{\epsilon}}{\left( \delta - \sqrt{2}\boldsymbol{\epsilon} \right)^2}.
    \end{equation}
\end{lemma}
\begin{proof}
    Theorem 2 in \cite{ng2001spectral}.
\end{proof}

\begin{lemma} \label{lemma: daviskahan}
    Let $M$ and $\Tilde{M}$ be the real symmetric matrices and the matrix $E = \Tilde{M} - M$ is the perturbation matrix. Let the denote the eigen vector decomposition of the matrices $M$ and $\Tilde{M}$ as follows.
    \begin{equation}
        M = \begin{bmatrix}
            U_1 & U2
        \end{bmatrix}
        \begin{bmatrix}
            \Lambda_1 & 0 \\
            0 & \Lambda_2
        \end{bmatrix}
        \begin{bmatrix}
            U_1^T \\
            U_2^T
        \end{bmatrix} \text{ and }
        \end{equation}
        \begin{equation}
        \Tilde{M} = \begin{bmatrix}
            \Tilde{U}_1 & \Tilde{U}2
        \end{bmatrix}
        \begin{bmatrix}
            \Tilde{\Lambda}_1 & 0 \\
            0 & \Tilde{\Lambda}_2
        \end{bmatrix}
        \begin{bmatrix}
            \Tilde{U}_1^T \\
            \Tilde{U}_2^T
        \end{bmatrix}.
    \end{equation}
    We assume that the spectra of $\Lambda_1$ and $\Tilde{\Lambda}_2$ is separated, that is, there exist a $\Delta>0$ such that the eigen values of $\Lambda_1$ lies entirely in $[\alpha, \beta]$ and the diagonal entries of $\Tilde{\Lambda}_2$ lies entirely outside $(\alpha-\Delta, \beta+\Delta)$ (or such that the eigen values of $\Lambda_1$ lies entirely outside $(\alpha-\Delta, \beta+\Delta)$ and the diagonal entries of $\Tilde{\Lambda}_2$ lies entirely in $[\alpha, \beta]$), and assume that the spectra of $\Lambda_2$ and $\Tilde{\Lambda}_1$ are also separated. Let $\Theta$ be the diagonal matrix with the diagonal entries being the sine of the canonical angles between the subspace spanned by the columns of $U_1$ and the subspace spanned by the columns of $\Tilde{U}_1$. Then for every unitary invariant norm, $\|\sin\left( \Theta \right)\| \leq \frac{\|E\|}{\Delta}$.
\end{lemma}
\begin{proof}
    Proposition 6.1 (Symmetric $\sin\theta$ Theorem) in \cite{davis1970rotation}.
\end{proof}

\begin{lemma} \label{lemma: weyl}
    Let $M$ and $\Tilde{M}$ be the real symmetric matrices of dimension $n\times n$ and the matrix $E = \Tilde{M} - M$ is the perturbation matrix. Let $\lambda_i(\cdot)$ represents the $i^{th}$ smallest eigen value of the matrix. Then we have the following inequality.
    \begin{equation}
        \left| \lambda_i\left(\Tilde{M}\right) - \lambda_i(M) \right| \leq \|E\|_2 \ \ \text{for all } i \in [n].
    \end{equation}
\end{lemma}
\begin{proof}
    Corollary 4.3.15 (Weyl's Theorem) in \cite{horn2012matrix}.
\end{proof}

\begin{lemma} \label{lemma: canonical}
    Let $X, Y \in \mathbb{C}^{M \times l}$ with $X^HX=I$ and $Y^HY = I$. If $2l\leq n$, there are unitary matrices $Q$, $U$ and $V$ such that
    \begin{equation*}
        QXU = \begin{array}{c}
            \\ l \\ l\\  n-2l
        \end{array}
        \begin{array}{c}
        l\\
        \begin{pmatrix}
            I\\0\\0
        \end{pmatrix}
        
        \end{array} \text{ and, }
        QYV = \begin{array}{c}
            \\ l \\ l\\  n-2l
        \end{array}
        \begin{array}{c}
        l\\
        \begin{pmatrix}
            \Gamma\\ \Sigma\\0
        \end{pmatrix}
        
        \end{array}
    \end{equation*}
    where, $\Gamma = \text{diag}(\gamma_1, \ldots, \gamma_l)$ and $\Sigma = \text{diag}(\sigma_1, \ldots, \sigma_l)$ satisfy $0\leq\gamma_1\leq \ldots \leq \gamma_l$, $\sigma_1\geq \sigma_2\geq \ldots \sigma_l\geq 0$, $\gamma_i^2+\sigma_i^2=1, i = 1, \ldots, l$.\\
    On the other hand, if $2l>n$, then $Q, U, V$ may be chosen so that
    \begin{equation*}
        QXU = \begin{array}{c}
            \\ n-l \\ 2l-n\\  n-l
        \end{array}
        \begin{array}{c}
        n-1 \ \ \  2l-n\\
        \begin{pmatrix}
            I &\ \ \ \ \ \ \ \ 0\\
            0 &\ \ \ \ \ \ \ \ I\\
            0 &\ \ \ \ \ \ \ \ 0
        \end{pmatrix}
        \end{array} 
    \text{ and, }
        QYV = \begin{array}{c}
            \\ n-l \\ 2l-n\\  n-l
        \end{array}
        \begin{array}{c}
        n-1 \ \ \  2l-n\\
        \begin{pmatrix}
            \Gamma &\ \ \ \ \ \ \ \ 0\\
            0 &\ \ \ \ \ \ \ \ I\\
            \Sigma &\ \ \ \ \ \ \ \ 0
        \end{pmatrix}
        \end{array} 
    \end{equation*}
     where, $\Gamma = \text{diag}(\gamma_1, \ldots, \gamma_{n-l})$ and $\Sigma = \text{diag}(\sigma_1, \ldots, \sigma_{n-l})$ satisfy $0\leq\gamma_1\leq \ldots \leq \gamma_{n-l}$, $\sigma_1\geq \sigma_2\geq \ldots \sigma_{n-l}\geq 0$, $\gamma_i^2+\sigma_i^2=1, i = 1, \ldots, n-l$.\\
     Furthermore, let $\mathcal{X}$ and $\mathcal{Y}$ be the subspaces spanned by the columns of the matrices $X$ and $Y$ respectively. Then the canonical angles between $\mathcal{X}$ and $\mathcal{Y}$ are the diagonal of the matrix $\Theta(X, Y)\coloneqq \arcsin\Sigma$.
\end{lemma}
\begin{proof}
    Theorem 5.2 and Definition 5.3 in \cite{stewart1990matrix}.
\end{proof}
\end{document}